\documentclass[acmsmall,screen]{acmart}
\newif\ifdraft\draftfalse   

\newif\ifappendix\appendixtrue   

\setcopyright{none}

\renewcommand\footnotetextcopyrightpermission[1]{}

\usepackage{bm,centernot,mathdots,mathpartir,mathtools,nicefrac,proof, amsmath}
\usepackage{microtype}
\usepackage{dsfont}
\usepackage[makeroom]{cancel}
\usepackage[frak=euler]{mathalfa}
\usepackage{fancyvrb}
\usepackage[plain]{algorithm}
\usepackage[noend]{algpseudocode}

\usepackage{relsize,xparse,xspace}
\usepackage{graphicx,transparent,wrapfig}

\usepackage{tikz}
\usetikzlibrary{arrows.meta}
\usetikzlibrary{decorations.pathmorphing}

\usepackage{caption}
\usepackage{subcaption}

\usepackage{booktabs, hhline, multirow, tabularx}

\usepackage[inline]{enumitem}
\usepackage{csvsimple}
\setlist{leftmargin=14pt}
\newcommand{\inliststyle}[1]{\textnormal{\textbf{\small#1}}}
\newlist{inlist}{enumerate*}{1}
\setlist[inlist]{label={\inliststyle{(\arabic*)}}}

\usepackage{multicol} 

\usepackage{listings}
\usepackage[listings]{tcolorbox}
\usepackage{accsupp}
\newcommand*{\noaccsupp}[1]{\BeginAccSupp{ActualText={}}#1\EndAccSupp{}}

\newcommand{\metagrammar}{$\mathcal{M}$}

\usepackage[nameinlink]{cleveref}
\crefname{algorithm}{Algorithm}{Algorithms}
\Crefname{algorithm}{Algorithm}{Algorithms}
\crefname{conjecture}{Conjecture}{Conjectures}
\Crefname{conjecture}{Conjecture}{Conjectures}
\crefname{definition}{Definition}{Definitions}
\Crefname{definition}{Definition}{Definitions}
\crefname{equation}{Equation}{Equations}
\Crefname{equation}{Equation}{Equations}
\crefname{figure}{Figure}{Figures}
\Crefname{figure}{Figure}{Figures}
\crefname{lemma}{Lemma}{Lemmas}
\Crefname{lemma}{Lemma}{Lemmas}
\crefname{section}{\S\!}{\S\S\!}
\Crefname{section}{\S\!}{\S\S\!}
\crefname{theorem}{Theorem}{Theorems}
\Crefname{theorem}{Theorem}{Theorems}

\definecolor{dkblue}{rgb}{0,0.1,0.5}
\definecolor{dkcyan}{rgb}{0.1, 0.3, 0.3}
\definecolor{dkgreen}{rgb}{0,0.3,0}
\definecolor{dkred}{rgb}{0.6,0,0}
\definecolor{dkpurple}{rgb}{0.7,0,0.4}
\definecolor{olive}{rgb}{0.4, 0.4, 0.0}
\definecolor{orange}{rgb}{0.9,0.6,0.2}

\definecolor{lightyellow}{RGB}{255, 255, 179}
\definecolor{lightgreen}{RGB}{170, 255, 220}
\definecolor{teal}{RGB}{141,211,199}
\definecolor{darkbrown}{RGB}{121,37,0}
\definecolor{princetonorange}{RGB}{255,143,0}
\definecolor{tmlblue}{RGB}{0,58,120}       

\newtcolorbox{simplebox}
             [1][]
             {arc=2pt,
              boxrule=0.75pt,
              boxsep=0.5pt,
              colback=white,
              left=8pt,
              right=8pt,
              top=0pt,
              bottom=5pt,
              #1,
             }

\newtcolorbox{inlinebox}
             [3][]
             {arc=0.5pt,
              boxrule=0.25pt,
              boxsep=0.5pt,
              colback=white,
              left=5pt,
              right=8pt,
              top=4pt,
              bottom=3pt,
              width=#2\linewidth,
              after skip=0.25em,
              valign=center,
              halign=center,
              #1,
             }
\let\oldimplies\implies
\renewcommand\implies{\mathrel{\resizebox{1.5em}{0.5em}{$\oldimplies$}}}
\let\oldiff\iff
\renewcommand\iff{\mathrel{\resizebox{1.75em}{0.5em}{$\oldiff$}}}

\newcommand{\CF}[1]{\text{\textnormal{\lstinline!#1!}}}
\newcommand{\GEq}{\ensuremath{::=~}}

\newcommand{\MG}{\ensuremath{MG}\xspace}
\newcommand{\Pref}{\ensuremath{\mathit{pref}}\xspace}
\newcommand{\Prod}{\ensuremath{\mathit{prod}}\xspace}
\newcommand{\Id}{\ensuremath{\mathcal{I}}\xspace}
\newcommand{\Reals}{\ensuremath{\mathds{R}}\xspace}

\newcommand{\Constraint}{\ensuremath{\mathit{c}}\xspace}
\newcommand{\SemanticsOf}[1]{\ensuremath{[ \! [#1] \! ]}}

\newcommand{\VARIADIC}[6]{%
  \expandafter\newcommand\csname GobbleNext#1Arg\endcsname[2]{%
    \csname CheckNext#1Arg\endcsname{##1#4##2}%
  }%
  \expandafter\newcommand\csname CheckNext#1Arg\endcsname[1]{%
    \csname @ifnextchar\endcsname\bgroup{\csname GobbleNext#1Arg\endcsname{##1}}{#2{##1#5}#6}%
  }%
  \expandafter\newcommand\csname #1\endcsname[2]{%
    \csname CheckNext#1Arg\endcsname{#3##1#4##2}%
  }%
}

\newcommand{\COMMENT}[3]{\ifdraft\textcolor{#1}{[\textbf{#2}:
\textnormal{#3}]}\fi}

\newcommand{\afm}[1]{\COMMENT{dkgreen}{afm}{#1}}

\newcommand{\dl}[1]{\COMMENT{dkred}{dl}{#1}}
\newcommand{\ksf}[1]{\COMMENT{dkpurple}{ksf}{#1}}
\newcommand{\rec}[0]{\textbf{R}}

\newcommand{\prov}[0]{\textbf{P}}

\newcommand{\OMIT}[1]{}

\newcommand{\R}{R} 
\newcommand{\client}{instance engineer\xspace}
\newcommand{\expert}{grammar engineer\xspace}

\newcommand{\Expert}{Grammar engineer\xspace}

\newcommand{\TOOL}[1]{\textnormal{\textproc{#1}}}
\newcommand{\Sys}{\TOOL{Saggitarius}\xspace}

\newcommand{\Glade}{\TOOL{Glade}}
\newcommand{\RFixer}{\TOOL{RFixer}}

\algblock{Match}{EndMatch}
\algnewcommand\algorithmicmatch{\textbf{match}}
\algnewcommand\algorithmicmatchwith{\textbf{with}}
\algnewcommand\algorithmicendmatch{\textbf{end\ match}}
\algrenewtext{Match}[1]{\algorithmicmatch\ #1\ \algorithmicmatchwith}
\algrenewtext{EndMatch}{\algorithmicendmatch}
\algtext*{EndMatch}

\algblock{Case}{EndCase}
\algnewcommand\algorithmiccase{|}
\algnewcommand\algorithmicendcase{}
\algrenewtext{Case}[2]{\algorithmiccase\ #1\ $\rightarrow$\ #2}
\algrenewtext{EndCase}{\algorithmicendcase{}}
\algtext*{EndCase}

\algnewcommand\Input{\textbf{input: }} 
\algnewcommand\Output{\textbf{output: }}

\newcommand{\VALUE}[1]{\textnormal{\texttt{#1}}}

\newcommand{\CounterExample}[1]{\VALUE{C\kern0.0625pcEx}\,#1}

\newcommand{\SysMode}{\textbf{\Sys}\xspace}
\newcommand{\ProsynthMode}{\textbf{ProSynth}\xspace}
\newcommand{\ProsynthPPMode}{\textbf{ProSynth++}\xspace}

\newcommand{\OPERATOR}[1]{\ensuremath{#1}}
\NewDocumentCommand
  {\Sufficient}
  { O{\mkern-2mu} O{} m }
  {\OPERATOR{\textnormal{\textsf{Suf}}^{\mkern3mu{#1}}_{#2}[#3]}}
\NewDocumentCommand
  {\Constructible}
  { O{} m m }
  {\OPERATOR{\bm{\mathfrak{C}}_{\mkern1mu#1}\left[#2\mkern3mu\bm{;}#3\right]}}
\newcommand{\EXPR}[1]{\ensuremath{#1}}
\VARIADIC{Apply}{\EXPR}{\left(}{\mkern6mu}{\right)}{}
\VARIADIC{Tuple}{\EXPR}{\left\langle}{,}{\right\rangle}{}
\newcommand{\TYPE}[1]{\ensuremath{#1}}

\newcommand{\HasType}[2]{\TYPE{#1 \bm{:} #2}}
\NewDocumentCommand
  {\HasTypeInCtx}
  { O{} m m }
  {\TYPE{#1 \vdash \HasType{#2}{#3}}}
\VARIADIC{TTuple}{\TYPE}{\left(}{\ast\,}{\right)}{}
\VARIADIC{TSum}{\TYPE}{\left(}{+}{\right)}{}
\VARIADIC{TFun}{\TYPE}{\left(}{\to}{\right)}{}

\newcommand{\relationRule}[4][]{\inferrule*[lab={\sc #2},#1]{#3}{#4}}

\newcommand{\rel}[1][m]{\ensuremath{\Rightarrow_{#1}}\xspace}
\newcommand{\reltwo}{\ensuremath{\Rightarrow}}

\newcommand{\idmap}[1]{\ensuremath{\lfloor#1\rfloor}}

\newcommand{\start}{\texttt{start}\xspace}
\newcommand{\id}{\texttt{id}\xspace}
\newcommand{\nt}[1]{\id\{#1\}\xspace}
\newcommand{\exts}{\texttt{exists}\xspace}
\newcommand{\cond}[3]{\texttt{if}\ #1\ \texttt{then}\ #2\ \texttt{as}\ #3\xspace}
\newcommand{\rb}{\texttt{rb}\xspace}
\newcommand{\ps}{\texttt{ps}\xspace}
\newcommand{\cc}{\texttt{c}\xspace}
\newcommand{\p}{\texttt{p}\xspace}

\usepackage{algorithm}
\usepackage[noend]{algpseudocode}

\algnewcommand\ALGOrithmicswitch{\textbf{switch}}
\algnewcommand\ALGOrithmicmatch{\textbf{match}}
\algnewcommand\ALGOrithmiccase{\textbf{case}}
\algnewcommand\ALGOrithmicwith{\textbf{with}}

\algdef{SE}[SWITCH]{Switch}{EndSwitch}[1]{\ALGOrithmicmatch\ #1\ \ALGOrithmicwith}{\ALGOrithmicend\ \ALGOrithmicswitch}%
\algdef{SE}[CASE]{Case}{EndCase}[1]{$|~$ #1 $\rightarrow$}{\ALGOrithmicend\ \ALGOrithmiccase}%
\algdef{SE}[CaseTwo]{CaseTwo}{EndCaseTwo}[2]{$|~$ #1 $\rightarrow$ #2}{\ALGOrithmicend\ \ALGOrithmiccase}%
\algtext*{EndSwitch}%
\algtext*{EndCase}%
\algtext*{EndCaseTwo}%
\algtext*{EndSecondCase}%

\newlength{\algorithmicindentlength}
\algrenewcommand\algorithmicindent{\algorithmicindentlength}

\algnewcommand{\IfThen}[2]{
  \State \algorithmicif\ #1\ \algorithmicthen\ #2}
\algnewcommand{\IfThenElse}[3]{
  \State \algorithmicif\ #1\ \algorithmicthen\ #2\ \algorithmicelse\ #3}

\algrenewcommand\alglinenumber[1]{\smaller \textcolor{darkgray}{\texttt{#1}}\hspace{0.5em}}
\algnewcommand{\LineComment}[2][1]{\State {\smaller\hspace{\dimexpr \algorithmicindentlength * #1 - 2.25em \relax}\textbf{\(\blacktriangleright\)\: #2}}}
\algrenewcommand{\algorithmiccomment}[1]{\hfill {\smaller \textbf{\(\blacktriangleright\)\: #1}}}

\algrenewcommand{\Return}[1]{\State \algorithmicreturn\ #1}

\algnewcommand\algorithmicbreak{\textbf{break}}
\algnewcommand{\Break}{\State \algorithmicbreak}

\algnewcommand\algorithmicthrow{\textbf{throw}}
\algnewcommand{\Throw}[1]{\State \algorithmicthrow\ #1}

\algnewcommand\algorithmicassume{\textbf{assume}}
\algnewcommand\Assume[1]{\State\algorithmicassume\ #1}
\algnewcommand\algorithmicassert{\textbf{assert}}
\algnewcommand\Assert[1]{\State\algorithmicassert\ #1}

\algblock{Func}{EndFunc}
\algnewcommand\algorithmicfunc{\textbf{func}}
\algnewcommand\algorithmicendfunc{\textbf{end\ func}}
\algrenewtext{Func}[2]{\algorithmicfunc\ \textsc{#1}\,(#2)}
\algrenewtext{EndFunc}{\algorithmicendfunc}
\algtext*{EndFunc}

\algblock{ParFor}{EndParFor}
\algnewcommand\algorithmicparfor{\textbf{parallel\,for}}
\algnewcommand\algorithmicpardo{\textbf{do}}
\algnewcommand\algorithmicendparfor{\textbf{end\ parallel\,for}}
\algrenewtext{ParFor}[1]{\algorithmicparfor\ #1\ \algorithmicpardo}
\algrenewtext{EndParFor}{\algorithmicendparfor}
\algtext*{EndParFor}

\algblock{ForEach}{EndForEach}
\algnewcommand\algorithmicforeach{\textbf{for\,each}}
\algnewcommand\algorithmicforeachdo{\textbf{do}}
\algnewcommand\algorithmicendforeach{\textbf{end\ for\,each}}
\algrenewtext{ForEach}[1]{\algorithmicforeach\ #1\ \algorithmicforeachdo}
\algrenewtext{EndForEach}{\algorithmicendforeach}
\algtext*{EndForEach}
\colorlet{listing-comment}{gray}
\colorlet{operator-symbol}{darkbrown}

\lstdefinelanguage{OCaml}{
    language=Caml,
    morekeywords={include, module, sig, struct, val},
    morekeywords=[2]{false, true},
    keywordstyle=[2]\color{dkgreen},
    morekeywords=[3]{int, bool, list},
    keywordstyle=[3]\color{dkcyan},
    literate=%
      {=}{{{\color{operator-symbol}=}}}1
      {<}{{{\color{operator-symbol}<}}}1
      {>}{{{\color{operator-symbol}>}}}1
      {:}{{{\color{operator-symbol}:}}}1
      {;}{{{\color{operator-symbol};}}}1
      {|}{{{\color{operator-symbol}|}}}1
      {[}{{{\color{operator-symbol}[}}}1
      {]}{{{\color{operator-symbol}]}}}1
      {\&}{{{\color{operator-symbol}\&}}}1
      {->}{{{\color{operator-symbol}->}}}1
}

\lstdefinestyle{default}{
    basicstyle=\linespread{0.9}\ttfamily,
    columns=fullflexible,
    commentstyle=\sffamily\color{black!50!white},
    escapechar=\#,
    framexleftmargin=1ex,
    framexrightmargin=1ex,
    keepspaces=true,
    keywordstyle=\color{dkblue},
    mathescape,
    numbers=left,
    numberblanklines=false,
    numbersep=1.25em,
    numberstyle=\relscale{0.8}\color{gray}\ttfamily\noaccsupp,
    showstringspaces=false,
    stepnumber=1,
    xleftmargin=2.5em,
}

\lstnewenvironment{lstlistingsmall}[1][]
  {\small\lstset{#1}}
  {}

\newtcblisting
    {boxed-listing}
    [2][]
    {listing only,
     arc=2pt,
     boxrule=0.75pt,
     boxsep=0.5pt,
     colback=white,
     left=-4pt,
     right=4pt,
     top=1pt,
     bottom=0pt,
     listing options={
         xleftmargin=5.5em,
         #2,
     },
     #1,
    }

\usepackage{balance}

\makeatletter
    \def\balanceissued{unbalanced}
    \let\oldbibitem\bibitem
    \def\bibitem{%
        \ifnum\thepage=\getrefnumber{TotPages}%
            \expandafter\ifx\expandafter\relax\balanceissued\relax\else%
                \balance%
                \gdef\balanceissued{\relax}\fi%
            \else\fi%
        \oldbibitem}
\makeatother

\makeatletter
\newcommand{\customlabel}[2]{%
   \protected@write \@auxout {}{\string \newlabel {#1}{{#2}{\thepage}{#2}{#1}{}} }%
   \hypertarget{#1}{#2}
}
\makeatother

\usepackage{mathpartir} 

\begin{document}

\title[Saggitarius]
      {Saggitarius:  A DSL for Specifying Grammatical Domains}


\begin{abstract}
Common data types like dates, addresses, phone numbers
and tables can have multiple textual representations, and many 
heavily-used languages, such as SQL, come in several dialects.
These variations can cause data to be misinterpreted, leading to
silent data corruption,
failure of data processing systems, or even
security vulnerabilities.  \Sys{} is a new language
and system designed to help programmers reason about the format of data, by describing \emph{grammatical domains}---that is, sets of context-free grammars that
describe the many possible representations of a datatype.
We describe the design of \Sys via example and
provide a relational semantics. We show how \Sys may be used to analyze a
data set: given example data, it uses
an algorithm based on semi-ring parsing and MaxSAT to
infer which grammar in a given domain best matches
that data.
We evaluate the effectiveness of the algorithm on a benchmark suite of 110 example problems, and we demonstrate
that our system typically returns a satisfying grammar within a few seconds with only a small number of examples. We also delve deeper into a more extensive case study on using \Sys for CSV dialect detection. 
Despite being general-purpose, we find that \Sys offers comparable results to hand-tuned, specialized tools; in the case of CSV, it infers grammars for 84\% of benchmarks within 60 seconds, and has comparable accuracy to custom-built dialect detection tools.
\end{abstract}

\author{Anders Miltner}
\affiliation{
  \institution{Simon Fraser University}
  \city{Burnaby}
  \state{BC}
  \country{Canada}                   
}
\email{miltner@cs.sfu.ca}

\author{Devon Loehr}
\affiliation{
  \institution{Princeton University}
  \city{Princeton}
  \state{NJ}
  \country{USA}                   
}
\email{dloehr@princeton.edu}

\author{Arnold Mong}
\affiliation{
  \institution{Princeton University}
  \city{Princeton}
  \state{NJ}
  \country{USA}                   
}
\email{among@alumni.princeton.edu}

\author{Kathleen Fisher}
\affiliation{
  \institution{Tufts University}
  \city{Medford}
  \state{MA}
  \country{USA}                   
}
\email{kfisher@cs.tufts.edu}

\author{David Walker}
\affiliation{
  \institution{Princeton University}
  \city{Princeton}
  \state{NJ}
  \country{USA}                   
}
\email{dpw@cs.princeton.edu}
%


%

\maketitle

\section{Introduction}%
\label{sec:introduction}

Data transformation, cleaning, and processing is a tedious and 
difficult task that stands in the way of getting important 
\emph{information} out of the raw text all around us. 
One particularly thorny problem is that the same information is often formatted differently when drawn from different sources. Correctly parsing even a single
data source in the face of such ambiguities can be difficult; if data is drawn from multiple sources, the chance of formatting inconsistencies and errors skyrockets.
This lack of standardization can have real costs---if a file is incorrectly parsed, it may lead to
processing system failures, silent data corruption 
that pollutes critical data bases, or even security vulnerabilities~\cite{safedocs}.  

\OMIT{
\begin{figure}
    \includegraphics[width=\textwidth]{figures/pads-timestamps.png}
    \caption{The PADS Project standard library functions for reading dates, times and timestamps.}
    \label{fig:pads-data}
\end{figure}
}

For example, consider processing a data file that contains a collection of dates---those dates can assume a range of different formats, and some of those formats are
ambiguous.  If a data processing system assumes a DD/MM/YY format when MM/DD/YY appear in the data, the system may crash or corrupt back-end databases.
Domain specific languages for data processing, such as PADS~\cite{pads}, aim
to help users solve such problems by providing a high-level language that
allows users to describe their data \emph{exactly}.  However, doing so is not without its own
difficulties.  For example,  the PADS Project’s standard library~\cite{padsmanual} provides 108 (!) different functions \emph{just for reading dates, times, and timestamps} (some due to differing
data encodings such as ASCII vs binary).
Other common basic types found in the ``ad hoc'' data formats processed by
systems like PADS include phone
numbers, addresses, times, names, postal codes,
and abbreviations (such as for states or provinces).  Such data also has variations,
especially when that data may involve different national conventions~\cite{wikiphone,wikipostal}.
One
can see why a programmer new to a system like PADS would have difficulty selecting among options to craft a format entirely by hand.  Indeed, even experts can face problems:  Fisher~\cite{pads-inference}
describes an experience working with PADS at AT\&T where it took roughly three
weeks to craft a description of a data format, in part because the format was not consistent, shifting part-way through.

\OMIT{
As an example, consider the myriad possible representations of a date.
Assuming a DD/MM/YY format when a MM/DD/YY format is used may cause a program to crash or produce incorrect results. While projects such as the PADS project~\cite{pads} are designed to simplify the writing of file format descriptions, it can be quite difficult to identify the data representation a given format uses when there so many possible representations. For example, the PADS Project's standard library~\cite{padsmanual} provides 108 different functions for reading dates, times, and timestamps, shown in Figure~\ref{fig:pads-data}.
Phone numbers, addresses, times, names, postal codes,
abbreviations (such as for states or provinces), and other common
datatypes suffer from these problems, particularly when the they are subject to different national conventions~\cite{wikiphone,wikipostal}.

The problems compound as formats get more complex.}

The problem is not limited to "ad hoc" formats like those addressed by the
PADS system.  For instance, one might think it safe to
assume that a “comma-separated-value” (CSV) file will be formatted as a series of values separated by commas.
For instance, one might think it safe to assume that a ``comma-separated-value'' (CSV) file will be
formatted as a series of \emph{values} separated by \emph{commas}. One would be wrong: tabs, vertical bars, semi-colons, spaces, or carats are all sometimes used instead of commas to separate
fields. Furthermore, some files use quotes to delimit strings and related symbols, while others use tildes or single quotes. And of course, some CSV files are ``typed,'' with, for example, one column expected to contain integers, and another strings. 

As
a result, the appropriate way to parse a CSV file can be ambiguous~\cite{csv-wrangling}; to avoid incorrectly interpreting 
it, one must determine the correct
CSV grammar (or ``dialect'') to use for a given data set. In the case of CSV, there are some tools to help users. For instance, Microsoft Excel provides a wizard that allows
a user to select the kind of delimiter to use when importing a CSV file as a spreadsheet; unfortunately, the process is 
a manual one. To try to deal with such variations in format, 
Python has libraries that implement ``sniffers'' to detect CSV dialects.  
However, the range of ambiguities make
such tools difficult to develop. 
Indeed, van den Burg \emph{et al.} provide a host of examples of the kinds of
ambiguities they found when analyzing CSV files in the wild.  In attempt to improve
the state of the art (just for CSV), they built their own custom CSV sniffer.  

CSV is not the only domain suffering problems due to ambiguity.
While the PDF format has been standardized, different tools implement different subsets (or even supersets!) of the standard.  Moreover, some of these PDF variants, and the tools used to process them, contain vulnerabilities~\cite{CarmonyHYBZ16,liu-blackhat-asia-2017}.  Hence,
ambiguities in how to process PDF lead not just to bugs, but possibly to
security vulnerabilities.  The DARPA SafeDocs program~\cite{safedocs}
is currently exploring ways to define safe subsets of PDF to limit
tool vulnerabilities.

The problem of determining the format of a file (or group of files) is called \emph{grammar induction}. Historically, grammar induction has been an exceptionally
difficult problem studied by countless researchers
(see~\cite{grammar-inference-survey,gold:inference,angluin:regular-sets-complexity,angluin:regular-sets,rivest:re-inference,garcia:k-testable,oncina-RPNI,regexp-positive-evidence} for
just a few examples from the literature).  Without any
prior information, the space of
grammars that might describe a data set is enormous.
Consequently, modern grammar induction tools either require
a large volume of examples, are specialized to particular tasks, and/or do not scale particularly
well. Indeed, even the restricted case of regular expression inference is a challenge:
Chen~\cite{chen2020multimodal} observed that
Angluin's classic L$^*$ algorithm makes 679 queries just to infer the simple regular expression [a-zA-Z]+,
and Lee~\cite{alpharegexp} shows that the search space for regular expressions
grows at a rate of $c^{2^d-1}$ where $d$ is the depth of the regular expression and $c$ is the number of regular expression
operators ($c$ is 7 when working with a binary alphabet).
While recent research~\cite{alpharegexp,chen2020multimodal}
has taken impressive steps forward in the special case of regular languages, scaling to larger data formats remains a
challenge.  For instance, AlphaRegex still classifies some regular expressions over a binary alphabet
with 10-20 symbols as ``hard''~\cite[Table 3]{alpharegexp}.  Scaling to richer classes of grammars, such as context-free languages, is even more daunting.

\paragraph{Syntax-Guided Grammar Induction} In this paper, we explore a new kind of grammar inference strategy, which we term \emph{syntax-guided grammar induction}. Our work was heavily inspired by another, closely related problem: program synthesis. 
Progress on that problem was made recently through
the use of syntax-guided (i.e., SYGUS~\cite{sygus}) methods and tools.
Just as SYGUS tools constrain programs
syntactically to cut down the search space and speed program synthesis, we constrain \emph{grammars} syntactically to cut down the search space and speed grammar synthesis.

Central to our strategy is a new way of thinking about the formatting options for a
datatype, centered around the
idea of a \emph{grammatical domain}. A grammatical domain is a \emph{set} of formal grammars
that describe the many possible representations of a datatype.
For instance, the \emph{date grammatical domain} would contain
one grammar describing the DD/MM/YY format, another grammar describing the MM/DD/YY format, and
many other grammars describing other formats (e.g., MM-DD-YY,
YYYY.MM.DD, and so on).

To rigorously specify and manage such domains, 
we develop a new language, called
\Sys.  Such a
language is a first step towards developing more robust data processing
tools. 
Roughly speaking, \Sys may be viewed as an extension of a standard YACC-based parser generator.  However, whereas YACC defines a single grammar and outputs a parser, \Sys defines a \emph{set} of possible grammars---\emph{i.e.,} a grammatical domain---and outputs a single grammar. 
To represent a set of grammars, \Sys allows \expert{}s to specify certain grammatical productions as optional.  Grammars within the domain are defined by the subset of
optional productions that they include.
\Sys also has features that allow \expert{}s
to declare constraints that force certain combinations of productions to appear, or not appear, 
and hence provides fine
control over the grammars of the domain in question.

Once a grammatical domain is defined in \Sys, it may be
\emph{applied} to example data. 
To apply a \Sys program, the user provides a set of example data files which are marked as either positive or negative. Then, our \emph{grammar induction algorithm} will attempt to find a matching element of the grammatical domain. That is,
it will select productions that allow the resulting grammar to parse all positive examples and none of the negative ones. Because more than one grammar from the domain may satisfy the provided examples, \Sys allows users to specify preferences that rank the generated grammars. For example, a preference might state that grammars with fewer productions should be preferred, or that certain productions are preferred over others.

We note that the ones \emph{defining}
the domain and those \emph{using} the domain may (and frequently will) be different. We call the person who defines the domain the 
\emph{grammar engineer} and the one who uses the grammar the
\emph{instance engineer}. These two roles are separate because
we expect heavy reuse of grammatical domains. For instance, the 
date domain need only be defined once by an expert. It can then
be used countless times by instance engineers developing specific
data processing tools for formats that contain dates.
While \expert{}s require sophisticated knowledge of a grammatical domain and the \Sys tool, \client{}s need only supply appropriate examples from the domain in question.

A key strength of the language and system design 
we propose is that it allows programmers to leverage their domain knowledge, which is often substantial. 
For example, if an engineer expects their data to be formatted roughly as CSV files containing dates, they can write or re-use CSV and date metagrammars to heavily restrict the search space.
In other words, we provide a linguistic framework that allows
programmers to apply their prior knowledge about the data,
and reduce the general grammar induction
problem to a more specialized and tractable one.

Furthermore, our approach leads to a
simple-yet-effective algorithm for grammar
induction that nonetheless requires very few examples. 
We study
the performance of our grammar induction algorithm
on a set of ten different induction tasks, considering
the performance of the algorithm on each task under a
range of different conditions. We find that \Sys frequently induces grammars in under 10 seconds, and
we demonstrate that \Sys can often learn the expected grammar with a mere handful of examples.

\paragraph{Use Cases}
\Sys is a broad tool, with a variety of potential applications. For example,
\begin{itemize}
    \item Alice maintains a large cloud service that allows clients to access a database using SQL queries. She has found that a common source of errors is when users  copy-paste code meant for a different SQL dialect. Alice could use a metagrammar to represent the different dialects of SQL, and use \Sys to produce a useful, automated error message that informs the user what dialect they're using, so they can easily adjust their code.
    \item Bob is in charge of aggregating data from a variety of sources, all represented as CSV files. However, different sources used different tools to generate their data, so each group of files potentially comes in its own unique variation of the CSV format. Bob can use the CSV metagrammar and apply it to each group of files individually to pick out the important features (separators, delimiters, etc), which can then be passed to an automated CSV parsing tool.
    \item Charlie is an IT professional at a high-security government facility. They wish to integrate a new, third-party PDF processing tool into their office's toolchain, but are worried that vulnerabilities in the tools may leak classified information if given certain inputs~\cite{safedocs}.  Charlie could use the PDF metagrammar to create a "shield" by finding a PDF benchmark suite, determining which benchmarks cause failures, and using \Sys to infer a grammar which matches only the successful benchmarks. They could then automatically compare inputs to this grammar before invoking the tool, ensuring it is only run on "good" inputs.
\end{itemize}

\OMIT{ 
\begin{itemize}
    \item \textbf{Dialect Detection:} Programmers often wish to run code from a variety of sources, e.g. including snippets from the internet, or running user-submitted database queries. However, such code can go wrong if it is written for a different version (or dialect) of the language. \Sys can be used to detect which dialect a given piece of code uses, to ease integration or provide useful, automated error messages.
    \item \textbf{Feature Detection:} Some languages, like CSV, come in all sorts of forms, rather than a few well-defined dialects. For these languages, \Sys can be used to automatically infer the important features -- for example, the separators and delimiters of a CSV file. The file can then be easily parsed, or automatically integrated into a spreadsheet or database.
    \item \textbf{Shielding Inputs:} An admin in a security-sensitive job may want to use prebuilt tools (e.g. for loading PDFs), but these tools may have critical security flaws when given certain inputs. \Sys can be used to create a "shield" by first running the tool on a benchmark suite, then feeding the benchmarks to \Sys with the failures marked as negative examples. The result will be a grammar for "good" files which do not cause security errors; this grammar can be used to ensure that the tool is only run on "good" inputs.
\end{itemize}
}

In each of these cases, \Sys allows users to rapidly develop a fitting tool for the domain at hand. Rather than manually writing a dialect recognizer, or a CSV sniffer, one need only write a metagrammar expressing their knowledge of the domain. When metagrammar components for the task already exist (dates, phone numbers, addresses, etc) one's task is reduced further. In the same way that it is easier to write a YACC file than to write a parser implementation, it is easier to write a metagrammar than a grammar induction engine from scratch.

\OMIT{ 
Each of these problems could instead be handled using ad-hoc programs or specialized tools for each data format. However, \Sys , however, has the \emph{versatility} to handle all of them. This versatility has two dimensions: \Sys can be broadly applied \emph{within} a grammatical domain, using a single metagrammar to perform a variety of induction tasks on different instances of the same datatype. On the other hand, \Sys can be applied \emph{across} domains by creating metagrammars for each of them -- after all, there isn't a fundamental difference between identifying date formats and identifying SQL dialects. 
\dl{I wasn't really sure what we were talking about regarding "the same problem across domains". I edited it out for now, but we could put it back in if we have a clear example or something}
}

Finally, \Sys \emph{works} -- despite its generality, it can achieve comparable performance to specialized, hand-tuned tools. We demonstrate in our evaluation section that when \Sys is used to detect the format of CSV files, it produces similar results to custom-built CSV sniffers.

To summarize, this paper makes the following contributions.

\begin{itemize}
\item We introduce the concept of \emph{grammatical domains}
and demonstrate that a variety of grammatical domains exist in the wild.
\item We design the first user-facing language, \Sys, for specifying grammatical domains precisely, and supply
an algorithm for inferring a candidate grammar from
the domain.
\item We create a benchmark suite for evaluating 
syntax-guided grammar induction algorithms and
evaluate the performance of \Sys on that suite, showing that it often returns a satisfying grammar within a few seconds, using few examples.
\item We demonstrate practical use cases of \Sys through a case study of the CSV metagrammar. We compare \Sys's performance to custom-built CSV sniffers, and find it yields comparable results. Additionally, we provide two additional case studies in \ifappendix Appendix~\ref{sec:case-studies} \else the full version of the paper \fi on metagrammars for XML and SQL. In the first, we demonstrate how the structure of metagrammars can be used to control the search space, speeding induction tasks. In the second, we show how \Sys can be used in practice to provide helpful feedback to users of query engines.

\OMIT{ 
In the first, we develop an XML metagrammar and we show how \Sys's programmatic control
over the search space can be used to speed up synthesis
times for challenging tasks.  In another, we test the use of \Sys on the problem of dialect detection for CSV files. 
In the second case, our general purpose tool infers grammars for 84\% of the files in the case study.  On these grammars, it has comparable accuracy to custom-built dialect detection tools~\cite{csv-wrangling-github}. 
In the final case, we apply \Sys to detect SQL dialects which can vary between companies or even within the same company. We also discuss how it can be used in practice to give more helpful feedback to front-end query engine users and back-end software developers who are tasked with debugging failing queries.
}
\end{itemize}

\OMIT{
\emph{Grammar induction} is the long-studied problem of inferring 
a grammar from positive and negative example data~\cite{grammar-inference-survey,gold:inference,angluin:regular-sets-complexity,angluin:regular-sets,rivest:re-inference,garcia:k-testable,oncina-RPNI,regexp-positive-evidence}.
Despite its obvious
value, and much research effort over the years, instances of the
problem often remain intractable, requiring either a great deal of
example data or vast computational resources to solve.  
Indeed, even the restricted case of regular expression inference is a challenge:
Chen~\cite{chen2020multimodal} observed that
Angluin's classic L$^*$ algorithm makes 679 queries just to infer the simple regular expression [a-zA-Z]+,
and Lee~\cite{alpharegexp} shows that the search space for regular expressions
grows at a rate of $c^{2^d-1}$ where $d$ is the depth of the regular expression and $c$ is the number of regular expression
operators ($c$ is 7 when working with a binary alphabet).  While recent research~\cite{alpharegexp,chen2020multimodal}
has taken impressive steps forward in the special case of regular languages, scaling to larger data formats remains a
challenge.  For instance, AlphaRegex still classifies some regular expressions over a binary alphabet
with 10-20 symbols as ``hard''~\cite[Table 3]{alpharegexp}.  Scaling to richer classes of grammars, such as the context-free languages, is even more daunting.


One ``solution'' to this problem is to give up.  The grammar 
induction problem cannot be solved accurately in general---there are simply too many degrees of freedom.  But after having given up, one must ask:  What's next?

As with any artificial intelligence problem, one way to give shape to, or cut down, the search
space, and thereby render an intractable problem tractable, is to associate
some priors, constraints, or background knowledge with instances of the problem.  
For example, over the last decade
or so, great progress has been made in the closely related problem of
program synthesis not by tackling the problem in general, but by
picking out key subdomains in which to operate.  Typically, this selection works by defining a restrictive domain-specific language (DSL),
specializing the ranking functions for the programs in the DSL, and then learning a ranked list of DSL programs from examples or other kinds of specifications.  In this way, it has been possible to build extremely effective program synthesizers for 
spreadsheet transformations~\cite{flashfill}, database queries~\cite{synth-database}, tree
transformations~\cite{yag+:pldi16} and lenses~\cite{optician,synth-symm-lenses,synth-quotient} to name just a few.  Program induction of this
form is often referred to as \emph{Syntax-guided Program Synthesis} (SyGuS), because one of the key elements 
of such systems is the fact that the syntax of programs is carefully constrained, via the DSL, ahead of time.
Now, there are meta-languages like Prose~\cite{prose} to help experts build syntax-guided program synthesis systems.

In this paper, we investigate whether similar techniques may be applied productively to the problem of grammar induction.  In doing so, the natural first question was whether there were any \emph{grammatical domains} akin to the \emph{programming domains} where SyGuS has been effective.  Such grammatical domains
should include a set of related grammars, just as a programming domain contains a set of related programs.
Moreover, programmers should easily be able to identify when the grammar they need falls within the given domain so they are able to select the tool or library to apply to their problem.  

\paragraph*{Grammatical Domains}
A clear example of an interesting grammatical domain is the set of ``comma-separated-value'' formats, the \emph{CSV domain}.
While one mostly thinks of comma-separated value formats as simple tables, separated by, well, commas,
there is actually great variation in these formats~\cite{csv-wrangling}. 
Tabs, vertical bars, semi-colons, spaces, or carats are all sometimes used rather than commas to separate
fields.  Some files use quotes to delimit strings and related symbols, while others use tildes or single quotes. And of course, some CSV files are ``typed,'' with integers
required to appear in one column and perhaps strings in another.  One of the consequences
of this mess is that the appropriate way to parse a CSV file can be ambiguous~\cite{csv-wrangling}.
Hence, to avoid incorrectly interpreting a given CSV file, one must determine the correct
CSV grammar or ``dialect'' to use for a given data set.  There are some tools to help users with
this particular special case.  For instance, Microsoft Excel provides a wizard that allows
a user to select the kind of delimiter to use to import a CSV file as a spreadsheet.  Unfortunately, the process is 
a manual one. Python has libraries that implement ``sniffers'' to detect CSV dialects
and van den Burg \emph{et al.}~\cite{csv-wrangling} built their own tool to
analyze CSV files and infer the grammar best suited to parse it.

CSV formats are just the tip of the iceberg: there are many other grammatical domains that arise in practice.  Some are similar to CSV formats in that they provide a 
textual representation of a large, structured data set.  Examples include the set of possible JSON formats, 
XML formats, or S-expression formats.  In each case, there is a general specification of the layout, 
but the general case is often specialized in particular applications, to, for instance, constrain the
types of values or the schema that may appear.  

Another sort of grammatical domain is a domain associated with a data type that,
over time, has come to have multiple different textual representations.  Here, dates
are a good example.  Dates may be formatted as MM/DD/YY, DD/MM/YY,
or one of myriad different ways.  Since these different formats may be ambiguous,
it is necessary to identify and deploy the specific format used in a particular data set. 
Moreover, since dates often appear nested within other formats,
those formats are themselves ambiguous until the internal date formats are properly determined.  
Other examples of data types with a variety of textual representations or specializations include phone numbers, postal codes, street addresses, 
 the names of states or provinces that have multiple common abbreviations, timezones, times, date ranges, and so on.  Any kind of data whose representation needs to be adjusted to handle internationalization is a likely candidate for a grammatical domain.

Finally, this work is partly movtivated by the US DARPA SafeDocs program~\cite{safedocs}, which is currently exploring the \emph{PDF grammatical domain}:  While PDF has been standardized, different tools implement subsets or supersets of the
standard.  Moreover, some of these PDF variants, and the tools used to process them, contain 
vulnerabilities~\cite{CarmonyHYBZ16,liu-blackhat-asia-2017}.  Consequently, SafeDocs has a number of teams investigating ways to define
safe subsets of PDF.

We collect the various grammatical domains we have studied into a benchmark suite described in \S \ref{sec:experiments}.

\dl{I don't like this paragraph. I think it would be stronger if we positions \Sys as a method of specifying grammatical domains, which we already defined, and \emph{also} for inferring particular elements given examples. The SGI problem itself has been pretty strongly implied already.}
\paragraph*{Syntax-guided Grammar Induction}
A domain-specific grammar induction problem is specified via a set of grammars (the domain),
a set of positive examples and a set of negative examples.  A solution to such a problem is
a grammar drawn from the domain that parses all the positive examples and none of the negative 
examples.  A natural way to specify such a domain is to constrain the syntax of the grammars
that may be chosen, by, for instance, specifying the candidate productions available.
We call problems specified this way \emph{Syntax-guided Grammar Induction (SGI) Problems}.

In SGI, or in domain-specific grammar induction more broadly, there are two types of user.  The first kind of
user, the \emph{\expert{}s}, are responsible for defining 
grammatical domains.  
The second, the \emph{\client{}s}, infer grammars from within the defined grammatical domains by supplying positive and negative examples. Ideally, the domains produced by \expert{}s are heavily reused, making their investment of time
and energy well spent.  For instance, the domain of dates can be reused in the synthesis of any grammar for any data 
containing a date.  Likewise, there are many many CSV files, meaning such a domain, once defined, can
be used by many \client{}s who need only supply example data and need not possess the skills of \expert{}s.
}

\OMIT{
\paragraph*{Saggitarius}
We designed and implemented \Sys, the first user-facing language and system to provide support for defining syntax-guided grammar induction problems.  Roughly speaking, \Sys may be viewed as an extension of a standard YACC-based parser generator.  However, whereas YACC defines a single grammar, \Sys defines a \emph{set} of possible grammars---\emph{i.e.,} a grammatical domain. 
To do so, \Sys allows \expert{}s to specify certain grammatical productions as optional, much like a SyGuS language, such as Sketch~\cite{sketch-original,sketch-asplos,sketch-summary}, declares certain subexpressions of a program optional.  
\Sys also has features that allow \expert{}s
to declare constraints that force certain combinations of productions to appear, or not appear, 
and hence provides fine
control over the grammars in the grammatical domain in question.  

The \Sys grammar induction algorithm uses both the definition of the grammatical domain and the supplied examples to select productions that allow the resulting grammar to parse all positive examples and none of the negative ones.  Because more than one grammar from the domain may satisfy the provided examples, \Sys allows \expert{}s to provide functions to rank the generated grammars. \Expert{}s play a role similar to designers of a domain-specific
language in syntax-guided synthesis; \client{}s play a role similar to users of synthesized functions. While \expert{}s require sophisticated knowledge of a grammatical domain and the \Sys tool, \client{}s need only supply appropriate examples from the domain in question and then use the generated parser.


}

\section{Motivating Examples}%
\label{sec:a-motivating-example}

Grammatical domains appear in many different contexts.  In this section, we show how to use
\Sys to define two useful domains:  the domain of calendar dates and the domain of comma-separated-value (CSV) formats.

\subsection{Example 1: Calendar Dates}
Dates are formatted in many different ways. Because the various formats are ambiguous (causing confusion as to whether
one is reading a MM/DD or DD/MM format, for instance), date parsers must be specialized to a particular data set.
Said another way, dates formats form a natural grammatical domain, and different data sets adhere to different grammars
within that domain.  

\Sys programs specify a grammatical domain through the use of a \emph{metagrammar}, 
which is a set of 
\emph{candidate productions} (a.k.a. \emph{candidate rules}) together with (a) constraints 
that limit which combinations of productions may appear, and (b) preferences that rank the grammars, for breaking ties when multiple grammars are applicable.


The simplest \Sys components specify productions using a YACC-like syntax with the form \texttt{N -> RHS.}  Here,
\texttt{N} is a non-terminal and \texttt{RHS} is a regular expression over terminals and non-terminals.
For instance, to begin construction of our date grammatical domain, we can specify \texttt{Digit} and 
\texttt{Year} non-terminals as follows.

\begin{center}
\begin{BVerbatim}
  Digit -> ["0"-"9"].
  Year -> Digit Digit | Digit Digit Digit Digit.
\end{BVerbatim}
\end{center}

This first definition looks like a definition one might find in an ordinary grammar.  It states 
that \texttt{Year} can have either two or four digits.  The denotation of such a definition is a grammatical domain---in
this case, a grammatical domain (a set) containing exactly one grammar.   

\Sys is more interesting when one defines metagrammars that include optional productions. 
Optional productions are preceded by a ``?'' symbol. For instance, consider the following:
\begin{center}
\begin{BVerbatim}
  Digit -> ["0"-"9"].
  Year -> ? Digit Digit ? Digit Digit Digit Digit.
\end{BVerbatim}
\end{center}
The metagrammar above denotes a grammatical domain that includes four grammars:
\begin{itemize}
    \item one grammar in which Year has no productions,
    \item two grammars in which Year has one production, and
    \item one grammar in which Year has two productions.
\end{itemize}
To extract a single grammar from this set of four grammars, one supplies the \Sys grammar induction engine with
positive and negative example data.  If no grammar parses all the data as required, the grammar induction algorithm
will return ``no viable grammar.''

Continuing, consider the following specification for days. 
\begin{center}
\begin{BVerbatim}
  Day -> ? ["1" - "9"]
         ? "0" ["1" - "9"]
         | ["1" - "2"] Digit | "30" | "31".
\end{BVerbatim}
\end{center}
This metagrammar includes grammars for days ranging from 1 -> 31.  It allows single digit
days to be prefixed with a 0. However, it is natural to desire grammars that parse either
single-digit days or 0-prefixed-days, but not both. One way to specify such a constraint is as follows.
\begin{center}
\begin{BVerbatim}
  constraint(|productions(Day)| = 4).
\end{BVerbatim}
\end{center}
Here, the constraint specifies that the number of production rules for \texttt{Day} must be exactly 4. Since the three productions on the last row are always included, exactly one of the ? production candidates can be in the solution. Another option is to name productions and to use the names in constraints:
\begin{center}
\begin{BVerbatim}
  Day -> ? ["1" - "9"] as SingleDigitDays
         ? "0" ["1" - "9"] as ZeroPrefixDays
         | ...
  constraint(SingleDigitDays XOR ZeroPrefixDays).
\end{BVerbatim}
\end{center}
Here, we have given each of the rule candidates a name, which is used as an indicator variable in the 
constraint expression; each evaluates to \texttt{true} iff a given grammar includes that production. All constraints must be a boolean combination of indicator variables or integer comparisons; we have provided commonly-used operations (such as counting the number of productions) as syntactic sugar.

\begin{figure}[h] 
\begin{boxed-listing}[left=-28pt]{}
Sep -> ? ","  ? "/"  ? "-".
constraint(|Production(Sep)| = 1)

Digit -> ["0"-"9"].

Year -> ? Digit Digit
        ? Digit Digit Digit Digit.
constraint(|Productions(Year)| = 1)

Month ->  ? Digit
          ? "0" Digit
          | "10" | "11" | "12".
constraint(|Productions(Month)| = 2)

Day -> ? ["1" - "9"]
       ? "0" ["1" - "9"]
       | ["1" - "2"] Digit | "30" | "31".
constraint(|Productions(Day)| = 2).

Date -> ? Day   Sep Month Sep Year 
        ? Month Sep Day   Sep Year
        ? Year  Sep Month Sep Day
        ? Year  Sep Day   Sep Month.
constraint(|Productions(Date)| = 1).

start Date
\end{boxed-listing} 
\caption{Calendar Dates Metagrammar}
\label{fig:example-1-main}
\end{figure}

Figure \ref{fig:example-1-main} includes the rest of the (simplified) definition of the date format, adding
definitions for separators, months, and dates as a whole. Dates are also designated as the start symbol. The use of constraints is common. For instance, notice the grammar engineer who designed this particular format allowed for the possibility of several different separators, but required a single separator to be used consistently throughout a format.  Hence, while a date format may use "\texttt{-}" or "\texttt{/}" as a separator, it never uses both.

%
%


To extract a particular grammar from the domain, an \client{} will supply positive and/or negative
example data.
For example, one could supply U.S.-style dates \texttt{12/31/72} and \texttt{01/01/72}, marking them as positive examples. 
Having done so, the \Sys grammar induction algorithm might generate the example grammar 
presented in Figure~\ref{fig:example-1-solution}.  
In this case, the solution presented in
Figure~\ref{fig:example-1-solution} is unique.  
However, in other circumstances, multiple grammars might
satisfy the given examples.
To manage multiple
grammars, 
one can rank grammars (e.g. preferring those with fewer productions) and ask for the most preferred grammar
according to the ranking. This is explained further in Section~\ref{sec:preferences}.


While one might worry that a naive \client could supply insufficient data and thereby underconstrain the set of possible
solution grammars, such problems could likely be mitigated through a well-designed user interface that informs
a user when multiple solutions are possible and presents example data to the user, asking them to choose valid and
invalid instances of the format.

\begin{figure}[t]\small%
\begin{boxed-listing}[left=-28pt]{}
Sep -> "/".

Digit -> ["0"-"9"].

Year -> Digit Digit.

Month ->  "0" Digit | "10" | "11" | "12".

Day -> "0" ["1" - "9"] 
     | ["1" - "2"] Digit 
     | "30" | "31".

Date -> Month Sep Day Sep Year.

S -> Date
\end{boxed-listing} 
\caption{Date Solution Grammar}%
\label{fig:example-1-solution} 
\end{figure} 

\subsection{Example 2: CSV}
\label{sec:example-csv}
Dates, phone numbers, addresses and the like are simple data types with many textual representations.
\Sys is also capable of specifying domains that contain larger, aggregate data types.  
The domain of comma-separated-value (CSV) formats is a good example.

One challenge in specifying a CSV domain is that if we want the columns of the CSV format
to be ``typed'' -- one column must be integers, another strings or dates, for instance -- we need to consider
many, many potential grammar productions.  To facilitate construction of such metagrammars succinctly, we allow
\expert{s} to define indexed collections of productions.  For instance, suppose we would like to specify
a spreadsheet with three columns (numbered 0-2) where each column contains either only numbers or only strings.
We might define the ith Cell in each row as:
\begin{center}
\begin{BVerbatim}
  Cell{i : [0,2]} -> ? Number ? String.
  forall (i : [0,2])
    constraint(|Productions(Cell{i})| = 1)    
\end{BVerbatim}
\end{center}
This declaration defines three nonterminals simultaneously: \texttt{Cell\{0\}}, \texttt{Cell\{1\}}, and \texttt{Cell\{2\}} and provides the
same definition for each of them.  However, since each of \texttt{Cell\{0\}}, \texttt{Cell\{1\}}, and \texttt{Cell\{2\}} are separate non-terminals, the underlying inference
engine can specialize them independently based on the supplied data.  For instance, \texttt{Cell\{0\}} could be a string and
\texttt{Cell\{1\}} and \texttt{Cell\{2\}} might both be numbers.\footnote{Observant readers may worry that the characters
``12'' could be interpreted as either a number or a string if the definition of strings includes numbers.
We will explain how to create preferences to disambiguate momentarily.}. Constraints can refer to specific
indexed non-terminals as shown.  

The domain of 3-column CSVs is likely a rare one!  Fortunately, users can also declare collections of indexed non-terminals
with an arbitrary natural-number size (e.g. \texttt{Cell\{i:nat\}} -> ...).

It is possible to restrict the possible productions based on an index. Below, we define \texttt{Row\{i\}}, a non-terminal for a row containing cells \texttt{Cell\{0\}} through \texttt{Cell\{i\}}.
The use of normal context-free definitions allows \texttt{Row\{i\}} to refer to \texttt{Row\{i-1\}}.  Notice that separators (\texttt{Sep}) are not indexed, so that one separator definition that is used consistently throughout the entire grammar.
\begin{center}
\begin{BVerbatim}
  Row{i : nat} ->
    ? if (i = 0) then Cell{i}
    ? if (i > 0) then Row{i-1} Sep Cell{i}.
        
  Sep -> ? "," ? "|" ? ";".
    constraint(|Productions(Sep)| = 1).
    
  Table{i : nat} -> MyRow{i} ("\n" MyRow{i})*.
\end{BVerbatim}
\end{center}

\texttt{Row\{i\}} represents a single row with \texttt{i} \texttt{Cell}s, and similarly \texttt{Table\{i\}} uses the standard Kleene star to represent a table with an arbitrary numbers of rows of length \texttt{i} (we could equally well have
written the usual recursive, context-free definition instead). The only difficulty that remains is the fact that, while we expect each CSV file to have a fixed number of columns, we cannot know that number in advance. \Sys allows users to represent such unknown quantities by declaring \emph{existential variables} and using them in rules. For example, we might use an existential variable to define a CSV format by adding the following code:
\begin{center}
\begin{BVerbatim}
  exists rowlen : nat  
  start Table{rowlen}.
\end{BVerbatim}
\end{center}

This CSV specification represents grammars in which all rows have a single, fixed length (e.g. a grammar of 5-column CSV files). We might want a more flexible grammar that permits any row length up to some maximum, so long as all rows in each file are the same length. We could represent this using a \emph{rule comprehension}, as below:

\begin{center}
\begin{BVerbatim}
  exists maxlen : nat
  S -> [? Row{len} for len = 0 to maxlen].
  start S.
\end{BVerbatim}
\end{center}

Here, the brackets \verb+[? ... for ... to ...]+ act similarly to 
a list comprehension.  They define one rule 
(\texttt{? Row\{len\}}) for each possible length up to \texttt{maxlen}.

Figure~\ref{fig:example-2-flawed} presents our progress so far on defining a metagrammar for simple CSV formats.
At the top, we have ``imported'' a couple of useful non-terminal definitions---definitions for \texttt{String} and \texttt{Number}.
Users can write such definitions from scratch, but we have developed a modest library of them to facilitate
quick construction of parsers for ad hoc data formats.

\begin{figure}[!h]\small%
\begin{boxed-listing}[left=-28pt]{}
import String, Number

exists rowlen : nat.

S -> Table{rowlen}.

Table{len : nat} -> Row{len} ("\n" Row{len})*.

Sep -> ? ","  ? "|" ? ";".
  constraint(|Productions(Sep)| = 1).

Row{len : nat} ->
  ? if (len = 0) then Cell{len}
  ? if (len > 0) then Row{len-1} Sep Cell{len}.

Cell{i : nat} -> ? Number ? String.
forall (i:nat) :
  constraint(|Productions(Cell{i})| = 1)
\end{boxed-listing}
\caption{CSV Metagrammar, Version 1}%
\label{fig:example-2-flawed} 
\end{figure} 

\subsection{Ranking Grammars}
\label{sec:preferences}

Consider the following example data.
\begin{center}
\texttt{
\begin{tabular}{l}
0,1,15,Hello world! \\
1,2,23,Programming \\
0,3,-2,rocks!
\end{tabular}
}
\end{center}
A human would probably claim that the first three columns should contain \texttt{Number}s, and the last one contains \texttt{String}s.
However, if \texttt{Number}s can be \texttt{String}s then the column types could be \texttt{Number}/\texttt{String}/\texttt{String}/\texttt{String}, or some other combination.
Without guidance, an algorithm will not know how to choose between potential grammars.

\Sys allows users to steer the underlying grammar induction algorithm towards the grammar of choice
by expressing preferences for one grammar over another.  Such preferences are expressed
using \texttt{prefer} clauses, which have a similar syntax to \texttt{constraint} clauses.  \texttt{Prefer} clauses assign
integer \emph{weights} to boolean formulas; the ranking of a synthesized grammar is the sum of the weights of all satisfied boolean formulas. 

Figure \ref{fig:example-2-preferences} illustrates the use of preferences to force \Sys to
infer \texttt{Number} cells when possible and \texttt{String} cells otherwise. For each $i$, the rank of the grammar is increased by 1 if the production \texttt{Num\{i\}} is included in the grammar. Since we have constrained each \texttt{Cell\{i\}} to have only one production, this leads \Sys to 
choose the \texttt{Number} rule whenever possible.

\begin{figure}[t]\small%
\begin{boxed-listing}[left=-28pt]{}
Cell{i : nat} ->
  ? Number as Num{i}
  ? String as Str{i}.
forall (i:nat) :
  constraint (|productions(Cell{i}| = 1).
  
forall (i:nat) : prefer 1 Num{i}.
\end{boxed-listing}
\caption{A Metagrammar with Preferences}%
\label{fig:example-2-preferences} 
\end{figure}

\subsection{Emitting Grammar Information}
Sometimes, users will make use of the entire inferred grammar (e.g. using it to build a parser or recognizer). In others, they may simply care about certain \emph{features} of the grammar. For example, using a built-in CSV parsing library might require the user to describe the type of each column. We could configure \Sys to print such information based on boolean conditions using the \lstinline{emit} syntax illustrated in Figure \ref{fig:example-2-emit}. This allows \Sys to be used as either a grammar inducer or a classifier (or both!) as needed.

\begin{figure}[t]\small%
\begin{boxed-listing}[left=-28pt]{}
Cell{i : nat} ->
  ? Number as Num{i}
  ? String as Str{i}.
forall (i:nat) :
  constraint (|productions(Cell{i}| = 1).

forall (i:nat) : emit "Row"+str(i)+":Num" if Num{i}.
forall (i:nat) : emit "Row"+str(i)+":String" if Str{i}.
\end{boxed-listing}
\caption{Example Emit Statements}%
\label{fig:example-2-emit} 
\end{figure}


\OMIT{
\section{The Syntax-Guided Grammar Induction Problem}
\label{sec:the-grammar-induction-problem}

A context free grammar $G = (S,\Id,\Sigma,R)$ with start symbol $S$, non-terminals $\Id$, terminals
$\Sigma$, and productions (also called rules) $R$ is defined in the usual way.  A Metagrammar $\mathcal{M} = (S,\Id,\Sigma,\R,f,h)$
includes the following components.
\begin{itemize}
            \item $S$, the grammar start symbol
            \item $\Id$, the set of grammar non-terminals
            \item $\Sigma$, the set of grammar terminals
            \item $\R$, the finite set of possible candidate productions
            \item Boolean-valued constraint function $f$: $2^R \to bool$
            \item Real-valued preference function $h$: $2^R \to \mathbb{R}$
        \end{itemize}
We say that a grammar $G$ belongs to $\mathcal{M}$ (written $G \in \mathcal{M}$) 
when the start symbol of $G$ is the
start symbol of $\mathcal{M}$, the terminal and non-terminal symbols are the same in $G$
and $\mathcal{M}$, the productions of $G$ are a subset of the productions of $M$, and
$f(G.\R)$ is true (where $G.\R$ are the productions of $G$).

An instance of the syntax-guided grammar induction problem is specified via a triple
($\mathcal{M}$, $Ex{+}$, $Ex{-}$) where $\mathcal{M}$ is a metagrammar and 
$Ex{+}$ and $Ex{-}$ are positive and negative examples respectively.
The solution to an instance of the syntax-guided grammar induction problem is a 
grammar $G \in \mathcal{M}$ that contains all of the positive examples, none of the
negative examples and has maximal ranking.  In other words, for all $G' \in \mathcal{M}$, $h(G.\R) >= h(G'.\R)$.
}

\section{Formal Language}
\label{sec:formal}

In this section, we formally describe the metagrammar specification language of \Sys, and describe how to interpret it using a relational semantics. The formal syntax is laid out in Figure \ref{fig:formal-grammar}. We describe how to compile the full language introduced in \S\ref{sec:a-motivating-example} to this core calculus in \S\ref{subsec:syntactic-sugar}.




                    






\begin{figure}[h]
\centering
\begin{subfigure}[b]{0.55\textwidth}
\centering
\begin{BVerbatim}
Int Exp.     e  := nat | id | e - 1

Bool Exp.    b  := T | F | e = e 
                 | e < e  | b && b 
                 | ...

Nonterminal  N  := id{e}

Production   p  := N | str | p p

Cond. Prod.  c  := if b then p
                    
Rule body    rb := c | c rb

Existential  ex := exists id

Rule         R  := id{id} -> ex.rb

Start symbol S  := start N

Metagrammar  MG := S | ex; MG | R; MG
\end{BVerbatim}
\caption{The core grammar of \Sys}
\label{fig:formal-grammar}
\end{subfigure}
\hfill
\begin{subfigure}[b]{0.40\textwidth}
\centering
\begin{BVerbatim}
Production  p   := id | str 
                 | p p

Productions ps  := [] | p ps

Rule        R   := id -> ps

Start       S   := start id

Grammar     G   := S | R G





\end{BVerbatim}
\caption{Syntax of a concrete grammar}
\label{fig:grammar-grammar}
\end{subfigure}
\vspace{10pt}
\caption{}
\label{fig:grammar}
\end{figure}

Our definition contains three built-in symbols: identifiers (\id), natural numbers (\texttt{nat}), and strings (\texttt{str}). There are two types of expression: integer, on which the only allowed operation is subtracting 1, and boolean, which permit the standard boolean operations as well as integer comparisons.

Nonterminals \texttt{N} have the form described in \S\ref{sec:example-csv}: an identifier followed by a natural-number argument. For simplicity, we require each nonterminal to have exactly one argument, but it is straightforward to extend the system to cover any a number of arguments. Productions in \Sys are simply a sequence of terminals (i.e., constant strings) and nonterminals.

The body of each rule is a sequence of conditional productions, combining productions with boolean conditions. Unlike in \S\ref{sec:a-motivating-example}, which included both mandatory and optional productions (beginning with \texttt{|} and \texttt{?}, respectively), the core grammar only uses mandatory, conditional productions.

A rule declaration \texttt{\nt{\id'} -> exists \id''.rb} declares a nonterminal identifier \id, which may be referenced in the rest of the program, as well as a name \id' for its argument and a local existential variable \id'', both of which may be referenced only in the rule body. We omit the type annotations on \id' and \id'' for simplicity. Finally, a metagrammar is a list of rule and global existential variable declarations, terminating in a declaration of the start symbol.

An astute reader may notice that grammar in Figure \ref{fig:formal-grammar} does not make reference to the constraints or preferences mentioned in \S\ref{sec:a-motivating-example}. We encode all constraints using a combination of local and/or global existential variables, and guards on the conditional productions. 
User preferences can easily be formalized
as functions from sets of rules to metrics.  We do
not give such preference functions a formal syntax and
semantics in this section, but assume such
functions exist in the following section.  In practice, these
preference functions are easily specified through the syntax used in \S\ref{sec:preferences}.

\subsection{Interpreting Metagrammars}
\label{sec:formal-example}
A metagrammar represents a set of \emph{concrete} (or \emph{candidate}) grammars, which have the form described in Figure \ref{fig:grammar-grammar}. Intuitively, this set contains all possible combinations of rules, generated by all possible instantiations of each existential variable. We formally define which which grammars are denoted by a given metagrammar in the next section, via a relational semantics. However, it is illuminating to first consider a simple example metagrammar, and determine which grammars it corresponds to.

Consider the following stripped-down CSV grammar, which describes only a single row. Note that the variables $i$ and $y$ are unused, and would be omitted in an actual \Sys program.
\begin{center}
\begin{BVerbatim}
Cell{i} -> exists x. // i unused
  | if (x = 0) then String
  | if (x <> 0) then Number.

Row{j} -> exists y. // y unused
  | if (j = 0) then Cell{j}
  | if (j > 0) then Row{j-1} Sep Cell{j}.
  
exists len.
start Row{len}.
\end{BVerbatim}
\end{center}

Intuitively, obtaining a grammar from a metagrammar has three steps. Each step involves several \emph{choices}, and different choices will result in different (though possibly equivalent) grammars. The first step is straightforward: we choose a value for each global existential variable. 

The second step is to remove arguments to nonterminals by duplicating each nonterminal definition once for each possible argument. For example, we turn the definition of \texttt{Row\{i\}} into definitions for new nonterminals \texttt{Row\_0}, \texttt{Row\_1}, and so on. The choice made in this step is the number $m$ of times we copy each definition. If we chose $\texttt{len} = 2$ and $m = 3$, we would transform the above grammar into the following:
\begin{center}
\begin{BVerbatim}
Cell_0 -> exists x.
  | if (x = 0) then String
  | if (x <> 0) then Number.
Cell_1 -> // Same as Cell_0
Cell_2 -> // Same as Cell_0

Row_0 -> Cell_0
Row_1 -> Row_0 Sep Cell_1
Row_2 -> Row_1 Sep Cell_2

start Row_2.
\end{BVerbatim}
\end{center}

Notice that $m$ must be strictly greater than the value chosen for $\texttt{len}$, since \texttt{len} is used as an argument to \texttt{Row}. If we chose $\texttt{len} = m = 2$, we would end up with a nonexistent start symbol \texttt{Row\_2}. In general, $m$ must be strictly larger than any of the existential values chosen previously.

The final step is to choose values for each local existential variable; again, these values must be strictly less than $m$. In the example above, if we chose values $0, 2, 1$ for the existentials inside \texttt{Cell\_0}, \texttt{Cell\_1}, and \texttt{Cell\_2} respectively, we would end with the following concrete grammar:
\begin{center}
\begin{BVerbatim}
Cell_0 -> String
Cell_1 -> Number
Cell_2 -> Number

Row_0 -> Cell_0
Row_1 -> Row_0 Sep Cell_1
Row_2 -> Row_1 Sep Cell_2

start Row_2.
\end{BVerbatim}
\end{center}

By making every possible choice in each of these three steps, we generate every grammar that is denoted by the metagrammar. This process is formalized in the next section.

\subsection{Interpreting Metagrammars, but formally this time}
Building upon the intuition from the previous section, we formally define the grammars $G$ represented by a metagrammar $MG$ using a relation $MG \rel G$. As before, the parameter $m$ denotes the number of times each rule is copied. We can then define the set of grammars denoted by $MG$ as
$$\SemanticsOf{MG} = \bigcup_{m=1}^\infty \{G | MG \rel G\}.$$

The three rules of the \rel relation follow much the same steps as the previous section; the only difference is that we fix $m$ in advance, and use it as an upper bound for the value of \emph{all} existential variables, global or local. We make use of an injective function $\idmap{ \cdot  }$ to translate the
parameterized identifiers of the meta-grammar into
the non-parameterized nonterminal identifiers of the
candidate grammars. For example, one implementation of this function would yield $\idmap{\texttt{foo}\{0\}} = \texttt{foo\_0}$. We leave $\idmap{\nt{e}}$ undefined if $e$ is not an integer value, and lift the function $\idmap{\cdot}$ to productions $p$ by applying it individually to each identifier within $p$.

We define the relation \rel below. In the process, we define a subordinate relation $\rb \reltwo \ps$ which relates rule bodies \rb to production lists \ps. We use the $@$ operator to denote list concatenation, and use the notation $[n/\id]$ to denote capture-avoiding substitution of $n$ for $\id$.
Finally, we use $\equiv$ to denote equivalence of expressions according to normal boolean and arithmetic rules.

\begin{mathpar}
\relationRule{Start} {
    e \equiv n\\
    \id' = \idmap{\nt{n}}
} {
    \start\ \nt{e} \rel \start\ \id'
}

\relationRule{Exist} {
    n < m\\
    MG[n/\id] \rel G
} {
   \exts\ \id;\ MG \rel G
}

\relationRule{Unrolling} {
   n_0 \dots, n_{m-1} < m
   \hspace{30pt}
   MG \rel G\hspace{10pt}\\
   \rb[0/\id'][n_0/\id''] \reltwo ps_0
   \hspace{10pt}\cdots\hspace{12pt}
   \rb[m-1/\id'][n_{m-1}/\id''] \reltwo ps_{m-1}\\
   \id_0 = \idmap{\nt{0}}
   \hspace{10pt}\cdots\hspace{10pt}
   \id_{m-1} = \idmap{\nt{m-1}}\\
} {
   \nt{\id'} \rightarrow \exts\ \id''.\rb; MG \rel \id_0 \rightarrow ps_0; \dots; \id_{m-1} \rightarrow ps_{m-1}; G
}

\relationRule{Cond-True} {
    b \equiv \texttt{True}\\ 
    \p' = \idmap{\p}\\
    e \equiv n\\
    \id' = \idmap{\nt{n}}
} {
  \cond{b}{\p}{\nt{e}} \reltwo [\p']
}

\relationRule{Cond-False} {
    b \equiv \texttt{False}
} {
  \cond{b}{\p}{\nt{e}} \reltwo []
}

\relationRule{Body} {
   \rb \reltwo \ps\\
   \cc \reltwo \ps'
} {
   \cc\ \rb \reltwo \ps\ @\ \ps'
}

\end{mathpar}

The \texttt{START} rule is the simplest; it simply checks that $e$ is a value, and applies $\idmap{\cdot}$ to translate the metagrammar nonterminal to a concrete-grammar nonterminal. The \texttt{EXIST} rule is almost as simple; it chooses a particular value for an existential variable, and substitutes that value throughout the rest of the metagrammar.

Most of the complexity lies in the \texttt{UNROLLING} rule. For each argument $i$ from $0$ to $m$, we choose a value $n_i$ for the local existential variable, and substitute both the the index $i$ and $n_i$ into the rule body. For each substituted body, we use the $\reltwo$ relation to obtain the list of corresponding productions $\texttt{ps}_i$ and add the rule $\idmap{\nt{i}} \rightarrow \texttt{ps}_i$ to the output grammar.

The \reltwo\ relation is relatively simple: for each conditional production, we return either a singleton list \texttt{[p']} if the condition is true, or an empty list if the condition is false. For rule bodies with multiple productions, we concatenate these lists, ultimately resulting in a list containing exactly those productions whose conditions are true.

\subsection{Syntactic Sugar by Example}
\label{subsec:syntactic-sugar}

As demonstrated in \S\ref{sec:a-motivating-example}, \Sys includes a number of features to enable rapid development of complex metagrammars that do
not appear in the core calculus, such as optional productions, constraints, and regular expressions. We informally illustrate here how to encode each of these features in our core calculus.

\paragraph*{Optional Productions} Optional productions can be represented using local existential variables. For example, \lstinline{A -> ? B ? C} can be represented as the following
rule
\begin{center}
\begin{verbatim}
    A -> exists b.
         | if (b = 0 || b = 1) then B
         | if (b = 0 || b = 2) then C
\end{verbatim}
\end{center}

\noindent Notice that depending on the choice of 
\lstinline{b}, 
any combination of the two rules may be included in the grammar.

\paragraph*{Regular Expression Productions} We compile regular expression syntax to context-free
grammar rules in the usual way, by introducing new nonterminals and rules as appropriate.  
Collections of rules generated from the
right-hand side of a single meta-grammar rule can be included or excluded as a group using global existential variables. For example, \lstinline{A -> ? "b"*} may be
compiled as follows, where g and Temp are fresh variables:
\begin{center}
\begin{verbatim}
exists g.
A -> if (g = 0) then Temp.
Temp -> if (g = 0) then "".
Temp -> if (g = 0) then "b" Temp.
\end{verbatim}
\end{center}

\paragraph*{Constraints} The surface language of \Sys includes several different types of constraints, but all of them can essentially be encoding by enumerating the possible combinations of rules. For example, imagine we have the program 
\begin{center}
\begin{minipage}{.40\textwidth}
\begin{verbatim}
A -> ? B as r1
     ? C as r2
     
\end{verbatim}
\end{minipage}
\begin{minipage}{.45\textwidth}
\begin{verbatim}
D -> ? E as r3
     ? F as r4
     ? G as r5
\end{verbatim}
\end{minipage}
\end{center}
We could encode the constraints that \texttt{A} has exactly one production, and \texttt{D} has at least two productions by creating local existential variables \texttt{b1} and \texttt{b2}, and enumerating the possible instantiations:
\begin{center}
\begin{minipage}{.40\textwidth}
\begin{verbatim}
A -> exists b1.
     | if (b1 = 0)  then B
     | if (b1 <> 0) then C

\end{verbatim}
\end{minipage}
\begin{minipage}{.45\textwidth}
\begin{verbatim}
D -> exists b2.
     | if (b2 <> 0) then E
     | if (b2 <> 1) then F
     | if (b2 <> 2) then G
\end{verbatim}
\end{minipage}
\end{center}
For something more complicated, consider the constraint \texttt{r1 => (r3 || r5)} which specifies that either \texttt{r3} or \texttt{r5} must be included whenever \texttt{r1} is. We can represent this cross-rule constraint using a \emph{global} existential variable $b$, and again simply enumerate the cases:
\begin{center}
\begin{minipage}{.40\textwidth}
\begin{verbatim}
exists b.
A -> ? C
     | if (b < 3) then B

\end{verbatim}
\end{minipage}
\begin{minipage}{.45\textwidth}
\begin{verbatim}
D -> ? E ? F ? G
     | if (b = 0 || b = 1) then E
     | if (b = 0 || b = 2) then G
\end{verbatim}
\end{minipage}
\end{center}

\OMIT{
Real-world data formats oftentimes have repeated structure and \afm{I'm going to work on other stuff for now, as I think it is apparent how these are written}

$A \rightarrow ? B ? C$ can be
compiled to two productions, $A \rightarrow B$ and $A \rightarrow C$. Adding a
production that is always in the grammar, like in the definition of
\lstinline{Digit}, requires jointly adding the production
\lstinline{Digit -> ["0"-"9"]} and
a constraint requiring the production to be included.}

\OMIT{The only part of \Sys not covered by this semantics is our use of  ``\lstinline{infty}.''
The infinity production creates an infinite search space, with an infinite number
of rules. When solving a problem with \lstinline{infty}, we are solving a more general
 problem that uses SGI as a
subproblem. We describe this generalization in more detail in \S\ref{sec:implementation}.}

\OMIT{
\section{\Sys Core Calculus - Old}

\Sys provides users with a convenient syntax for writing metagrammars (\MG) 
that consists of a sequence of statements, where each
statement can be either a production, a constraint, or a preference. The formal syntax follows.
\begin{center}
  \begin{tabular}{@{}r@{\ }c@{\ }c@{\ }l@{}}
    \MG & \GEq{} &   & \Prod{}. \MG \\
        &        & | & const \Constraint{}. \MG \\
        &        & | & prefer \Pref{}. \MG\\
        &        & | & $\epsilon$ \\
    \Prod & \GEq{} & & $i \rightarrow p$ \\
    \Constraint & \GEq{} & & $i \rightarrow p$ \hspace*{1em} |\ $\neg\Constraint$\\
      &        & | & $\Constraint_1 \wedge \Constraint_2$ 
      \hspace*{.96em}|\ $\Constraint_1 \vee \Constraint_2$ \\
    \Pref & \GEq{} & & $c$ $r$
  \end{tabular}
\end{center}

Productions in a \Sys metagrammar are of the form $i \rightarrow p$, where $i$
is a nonterminal, and $p$ is a sequence of terminals and non-terminals.

Constraints are boolean formulas, where the predicates are productions. Any
accepting grammar must satisfy the boolean formula, where included productions
are interpreted as $\textit{true}$ and productions not included are interpreted as $\textit{false}$.

Preferences consist of a boolean formula $c$ and a real value $r$. If the boolean
formula $c$ is satisfied by a given grammar, the reward of that grammar is increased
by $r$.

We assume there is some designated set of terminals, non-terminals and start symbol.
The semantics of a \Sys metagrammar is given by a function $\SemanticsOf{\cdot}$ that maps a metagrammar $\MG$ to a triple
$(R,f,h)$ of productions, constraint function and preference function. 
The
function $\mathcal{C}\SemanticsOf{\Constraint}$ 
maps a constraint $\Constraint$ to a function $f$ from
grammars to type $bool$. The function
$\mathcal{P}\SemanticsOf{\Pref}$ maps a preference $\Pref$ to a function from grammars to \Reals. 
We formally define these three functions below.

\begin{center}
  \begin{tabular}{@{}r@{\ }c@{\ }l@{\ }l@{}}
    $\SemanticsOf{\Prod{}. \MG}$ & = & $(\R \cup \{\Prod\},f,h)$\\
    & & \hspace*{2em}where $\SemanticsOf{\MG} = (\R,f,h)$\\
    $\SemanticsOf{\text{const }\Constraint. \MG}$ & = & $(\R,f,h)$\\
    & & \hspace*{2em}where $\SemanticsOf{\MG} = (\R,f_1,h)$\\
    & & \hspace*{2em}and $\mathcal{C}\SemanticsOf{\Constraint} = f_2$\\
    & & \hspace*{2em}and $f(x) = f_1(x) \wedge f_2(x)$\\
    $\SemanticsOf{\text{prefer }\Pref. \MG}$ & = & $(\R,f,h)$\\
    & & \hspace*{2em}where $\SemanticsOf{\MG} = (\R,f,h_1)$\\
    & & \hspace*{2em}and $\mathcal{P}\SemanticsOf{\Pref} = h_2$\\
    & & \hspace*{2em}and $h(x) = h_1(x) + h_2(x)$\\
    $\SemanticsOf{\epsilon}$ & = & $(\{\},f,h)$\\
    & & \hspace*{2em}where $f(x) = true$\\
    & & \hspace*{2em}and $h(x) = 0$\\
    $\mathcal{C}\SemanticsOf{i \rightarrow p}$ & = & $f$ where $f(x) = i \rightarrow p \in x$\\
    $\mathcal{C}\SemanticsOf{\neg \Constraint}$ & = & $f$ where $f(x) = \neg(\mathcal{C}\SemanticsOf{\Constraint}(x))$\\
    $\mathcal{C}\SemanticsOf{\Constraint_1\wedge \Constraint_2}$ & = & $f$ where $f(x) = \mathcal{C}\SemanticsOf{\Constraint_1}(x) \wedge \mathcal{C}\SemanticsOf{\Constraint_2}(x)$\\
    $\mathcal{C}\SemanticsOf{\Constraint_1\vee \Constraint_2}$ & = & $f$ where $f(x) = \mathcal{C}\SemanticsOf{\Constraint_1}(x) \vee \mathcal{C}\SemanticsOf{\Constraint_2}(x)$\\
    $\mathcal{P}\SemanticsOf{\Constraint\ r}$ & = & $f$ where $f(x) = $ if $\mathcal{C}\SemanticsOf{\Constraint}(x)$ then $r$ else 0\\
  \end{tabular}
\end{center}
}

\section{Grammar Induction Algorithm}%
\label{sec:grammar-induction-algorithms}
The goal of \Sys' grammar induction algorithm is, given a set of positive and negative examples, a metagrammar $MG$ and a ranking function $p: \SemanticsOf{MG} \rightarrow \mathbb{Z}$, to return the highest-ranked grammar contained in $\SemanticsOf{MG}$ which parses all the positive examples and none of the negative examples\footnote{If several of these grammars have the same rank, any of them may be returned.}. We have experimented with several possible algorithms that interleave parsing and
constraint solving in different ways, and outline the most successful algorithm here.

Our grammar induction algorithm is formalized in Algorithm~\ref{alg:grammarinduction}. At a high level, our strategy is rather simple. First, we concretize the metagrammar, removing arguments to nonterminals along with existential variables. Next, we determine all possible combinations of rules that successfully parse each example. We use these to create a logical formula asserting either that one of these combinations is included in the final grammar (for a positive example), or that none are included (for a negative example). For space, we elide the formal definitions of \textsc{Concretize} and \textsc{Constraints} and the proof of the following theorem. The interested reader can find them in \ifappendix Appendix \ref{sec:proofs}\else the full version of the paper\fi.

\begin{theorem}
If there is a Grammar $G$ such that $G \in \SemanticsOf{MG}$, where $Ex^+ \subseteq \mathcal{L}(G)$ and $Ex^- \subseteq \overline{\mathcal{L}(G)}$, then $\Call{InduceGrammar}{MG,Ex^+,Ex^-,0}$ will return such a grammar.
\end{theorem}

In \Sys, the ranking function $p$ is constructed piecewise from various \texttt{prefer} statements in the surface language, each consisting of an integer weight and a boolean combination of rules. The rank of a grammar is simply the sum of all weights whose boolean combinations are satisfied. These preferences are translated into logical formulas, combined with the formulas for each example, and shipped to Z3~\cite{z3}, which uses its MaxSMT algorithms to find a preference-optimal solution.

For simplicity, we assume there are no variables in our start symbol.
\begin{algorithm}[!t]
\begin{algorithmic}[1]
\Procedure{InduceGrammar}{$MG,Ex^+,Ex^-,k$}
\State $G \gets \Call{Concretize}{MG,k}$
\State $\phi \gets \Call{Constraints}{MG,k}$
\State $\phi^+ \gets \bigwedge_{ex \in Ex^+} \mathsf{SemiringParse}(G,ex)$
\State $\phi^- \gets \bigwedge_{ex \in Ex^+} \neg \mathsf{SemiringParse}(G,ex)$
\State $ret \gets \Call{SMT}{\phi^+ \wedge \phi^- \wedge \phi}$
\Match{$ret$}
\Case{Some($vs$)}{ \textbf{return }$MG[vs]$}
\EndCase
\Case{None}{ \textbf{return }$\Call{InduceGrammar}{MG,Ex^+,Ex^-,k+1}$}
\EndCase
\EndMatch
\EndProcedure
\end{algorithmic}
\caption{Syntax-Guided Grammar Induction Algorithm. $MG$ is the provided metagrammar, $Ex^+$ is the provided positive examples, $Ex^-$ is the provided negative examples, and $k$ is the search depth (initially 0). }
\label{alg:grammarinduction}
\end{algorithm}

\subsection{Concretizing Metagrammars}
Before we can apply a conventional parsing algorithm to determine which combinations of rules parse a given example, we must first eliminate \Sys' unconventional constructs -- namely, nonterminals with arguments, and existential variables. This process closely mirrors the one described in \S\ref{sec:formal-example}. We begin by choosing a maximum value for each global variable (corresponding to a particular choice of $m$). We then copy each rule that number of times, and expand out \emph{all possible} values for each global and local existential variable. 

This process generates a finite subset of $\SemanticsOf{MG}$, obtained by taking the union only up to our chosen value of $m$. We then apply the concrete parsing algorithm described below; if it fails to find a suitable grammar, we begin the process anew with a larger value of $m$. In general, this process is not guaranteed to terminate, but termination \emph{can} be guaranteed for certain types of grammar.

\paragraph{Optimizations} Our construction of $\SemanticsOf{MG}$ contains significant redundancy. For one, different values of existential variables may not actually lead to different grammars. \OMIT{We have seen cases of this already; in the encoding of \texttt{A -> ? B ? C} given in \S\ref{sec:syntactic-sugar}\dl{Rework this if we cut that section}, any value for variable $b$ larger than 2 produces the same grammar.} Similarly, we need not copy \emph{every} rule the same number of times, since this may result in rules which are unreachable from the start symbol. 

\Sys leverages these observations by allowing users to declare variables with a range type instead of type \texttt{nat}; for example, the line \texttt{exists x : [0, 4]} declares an existential variable whose value is constrained between 0 and 4, inclusive. By using these types, users can both express their intent for variables to be restricted, and provide the compiler with hints to avoid checking redundant grammars. Furthermore, \Sys copies rules only as much as necessary, ensuring that no unreachable rules appear in the concretized metagrammar.

Notice that metagrammars in which all variables are bounded are \emph{finite}, and hence permit an exhaustive search of candidate grammars. This allows us to strengthen our termination conditions: \textbf{Induction for \Sys metagrammars is guaranteed to terminate whenever all variables in the program are bounded.}

\subsection{Grammar Induction Via Semiring Parsing}

Once we have generated a large, ambiguous grammar from a metagrammar, we must identify a subset of rules in that grammar that parse all positive examples, do not parse any negative examples, and satisfy the user-provided constraints.

For example, consider the following grammar:
\begin{lstlisting}
{x_1} X -> x
{x_2} X -> x
{y_1} Y -> y
{y_2} Y -> y
{z_1} Z -> z
{z_3} Z -> z

G -> XG
G -> YG
G -> Z
\end{lstlisting}

Consider the example string "xyz". This could be parsed using rules $x_1, y_1, z_1$ or $x_1, y_1, z_2$ or $x_1, y_2, z_1$, and so on. A na\"ive algorithm could generate all combination of rules, filter out those that do not satisfy the user-provided constraints, and iteratively test them until the it finds a set of rules that accept the positive examples but reject the negative examples. Unfortunately such an approach is not tractable in practice -- as there are can be an exponential number of different rule combinations for parsing the same input string. Even in this simple example there are $2^3$ possible combinations for only a single string.

We address this problem by representing the possible rule combinations as a logical formula. For the example string above, the possible rule sets are represented by the formula $(x_1 \vee x_2) \wedge (y_1 \vee y_2) \wedge (z_1 \vee z_2)$. Rule names $x_1\ldots z_1$ are interpreted as boolean variables; satisfying assignments correspond to collections of those rules whose variable is \texttt{true}. Any satisfying assignment for this formula will yield a set of rules which suffice to parse the example.

More broadly, we encode the possible parse trees for every example as such logical formulas. We then create a single formula as a conjunction of the formulas for each positive example, the negation of the parse formulas for each negative examples, and the user-provided constraints. Finally, we pass this formula to an off-the-shelf SAT solver. As with many practical problems, SAT solvers can typically find satisfying assignments for these formulas in much less than exponential time.

\paragraph*{Generating The Formulas} The process of generating formulas can be implemented using semiring parsing~\cite{semiringparsing,EarleyDerivationTree}. In particular, we consider logical formulas to be elements in a conditional-table semiring~\cite{provsemi}, where logical disjunctions and conjunctions correspond to the semiring's $+$ and $\times$ operators.

\paragraph*{Satisfying the Formulas} Once provided to an SMT solver, the exponential number of possible rules does not diminish. However, this is the bread-and-butter of an SMT solver. In the past, most DNF formulas would incur an exponential blowup when given to an SMT solver; however, modern SMT solvers can efficiently find satisfying assignments for such formulas.

\begin{figure*}[t!]
  \centering
  \includegraphics[width=\columnwidth]{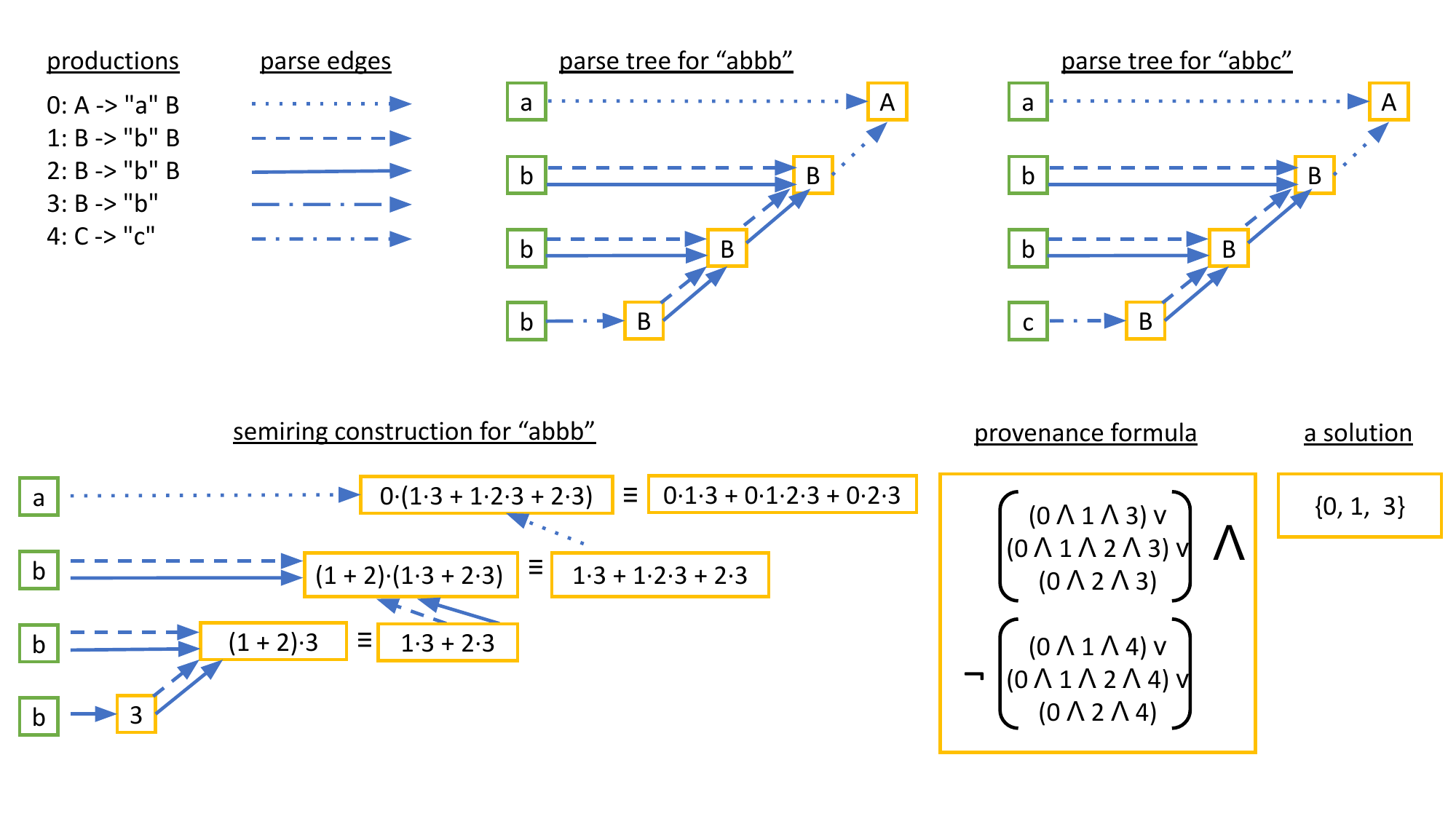}
  \caption{Semiring parsing and provenance formula construction}
  \label{fig:provenance-example}
\end{figure*}

Figure~\ref{fig:provenance-example} illustrates our process on a simple example.  In the upper left,
we present five candidate rules (numbered 0-4) that define nonterminals A, B and C. Rules 1 and 2 are identical in this example; they
help illustrate the ambiguity the algorithm must handle.  
In the upper right, Figure~\ref{fig:provenance-example} presents two example parse trees, one
for the example ``abbb'', and one for ``abbc''. Each edge of the parse tree is associated with a particular rule; for instance, rule 1 is denoted by a dashed edge and rule 2 is denoted by a solid edge. 

On the lower left-hand side,
Figure~\ref{fig:provenance-example} presents the strategy that the parser implements to compute a semiring expression for the example "abbb" (the computation for "abbc" is similar). These expressions represent the possible sets of rules that can parse each example. To begin, the final "b" of "abbb" can only be obtained via rule 3, which corresponds to the algebraic expression $3$. The second-to-last "b" character could be obtained either with rule 1 or rule 2, represented as the expression $1+2$. Hence, to produce a string ending in "bb", we must have either rule 1 or 2, and also rule 3. We can represent this by conjoining our two expressions to get $(1+2)\cdot3$, or equivalently $1\cdot2 + 1\cdot3$.

Continuing on, the third-to-last "b" could also be obtained via rule 1 or rule 2, so we once again conjoin the expression $1+2$ with our accumulator expression to get $(1+2)\cdot(1\cdot3+2\cdot3)$. Since the conditional table semiring is commutative and idempotent\footnote{Intuitively, the expression $1\cdot1\cdot3$ represents the set containing rule 1, rule 1, and rule 3; this is of course equal to the set containing rule 1 and rule 3, represented by $1\cdot3$. For similar reasons, $1\cdot2\cdot3 \equiv 2\cdot1\cdot3$.}, this is equivalent to $1\cdot3 + 1\cdot2\cdot3 + 1\cdot2$. Finally, the first "a" character could only be obtained from the start symbol \texttt{A} using rule 0, so we end up with a final expression of $0\cdot(1\cdot3+1\cdot2\cdot3+2\cdot3)$, or equivalently $0\cdot1\cdot3+0\cdot1\cdot2\cdot3+0\cdot2\cdot3$.

To finish the process, we must translate our example semiring expressions into logical formulas. Fortunately, this is 
easy---our intuition for the semiring expressions already told us that $+$ corresponds to disjunction and $\cdot$ to conjunction. Simply replacing the connectives transforms our semiring expression into a logical one. 

If we take "abbb" to be a positive example and "abbc" to be a negative example, we can obtain an overall provenance formula by translating both examples into semiring expressions, then to logical formulas, and conjoining them. Since "abbc" is negative, we negate the corresponding conjunct, indicating that \emph{none} of the sets of rules that can parse "abbc" should be included. 
We then pass the provenance formula to an SMT solver (specifically, Z3), to obtain a solution; in this case, one possible solution is the set of rules $\{0,1,3\}$. If we had included preferences in our example, we would use Z3's MaxSMT algorithms to obtain the solution with the highest rank.

\OMIT{
Here, the non-terminals are replaced by algebraic expressions that represent the rules used to
implement that subtree.  For instance, the orange box at the bottom left contains the algebraic 
expression ``$3$,'' which represents the fact that rule $3$ must be used to parse the underlying string.
The box above it contains the expression ``$1.3 + 2.3$,'' which indicates that either the rules
$1$ and $3$ or the rules $2$ and $3$ must be used to parse ``bb.''  Above
that, there appears another provenance expression, which may be reduced to the equivalent 
expression ``$1.3 + 1.2.3 + 2.3$'' thanks to the fact we are working in the conditional table semiring.
The final expression produced is ``$0.1.3 + 0.1.2.3 + 0.2.3$.''  
}

%

\OMIT{
\subsection{Notation}
\label{subsec:Notation}
\afm{I think this should go earlier in the paper, but I will leave here for
  now.} We introduce the following notation:
\begin{itemize}
    \item \metagrammar: Metagrammar 
    \item $\mathcal{M}|X$: Grammar achieved when restricting $\mathcal{M}$ to the candidate rules in $X$
    \item $G$: Context free grammar
    \item $\neg G$: All candidate rules not found in $G$
    \item $Ex+, Ex-, E$: Positive examples, negative examples, all examples $(E = Ex+ \cup Ex-)$
    \item $\phi$: SAT formula of candidate rule indicator literals whose solutions are grammars in \metagrammar \\
          $\Phi$: Set of SAT formulas whose solutions are in \metagrammar
    \item $h$: Real valued function over \metagrammar\\
          $H$: Set of all real valued functions over \metagrammar
    \item RC: Rule collection.
    \item prov: alias for rule collection with more semantic meaning (see Nit in this section)
\end{itemize}
\ksf{I think this section should explain how we encode a set of possible grammars as a SAT formula.}

While SyGuS and SGI have a lot of motivation in common, their search domains are
quite different. In SyGuS, one builds a tree, where the nodes of the tree are
selected from a number of options. In SGI, one makes a decision on what rules to
include from a set of candidate rules. This search space is similar to that of
logic program synthesis\cite{?,?}, where one chooses a set of candidate logic
programming rules from a space containing many possible rules. Indeed, as logic
programming can parse files from arbitrary grammars~\cite{?}, we could just use
existing logic program inducers to solve SGI.

\input{figures/algo-SATHeavy} However, doing so would not exploit properties of
grammars that aren't present in general logic programs. A number of parsing
algorithms permit ``parsing over semirings,'' where the parser can be
parameterized by a semiring, and will infer keep track of parsing information
according to that semiring. For example, a recognizer would operate over the
semiring \{0,1\}, where 0 is the additive identity, and 1 is the multiplicative
identity. Another semiring is the free semiring, which keeps track of all
possible parses of a given string. We instantiate a provenance
semiring~\cite{provsemi} to keep track of what rules are required to generate a
valid parse. From a high level, our algorithm operates in two steps:
\begin{enumerate}
\item Find the provenance formulas for all positive and negative examples.
\item Conjunct those provenance formulas with the formulas expressing
  constraints, and call a MaxSMT solver with the conjuncted formula and the
  preferences to find the optimal set of rules.
\end{enumerate}
These steps are shown formally in Algorithm~\ref{fig:algo-satheavy}.
Algorithm~\ref{fig:algo-satheavy} relies on two additional inputs: a SMT solver
and a full parser. We describe the instantiation of these two inputs in
\S\ref{sec:implementation}.
}

\OMIT{
\section{Implementation}%
\label{sec:implementation}


We have implemented \Sys in 3490lines of OCaml code. In this section, we
describe our implementation of our provenance parser and SMT-Solver integration.

\subsection{Provenance Parser} We have implemented a our Provenance Parser as a altered Earley parser. The construction of our Provenance Parser hinges on two primary pieces of prior work: (1) Parsing over semirings~\cite{parsingoversemirings,insideoutparsing}, (2) the conditional tables provenance semiring~\cite{valsstuff}.

Goodman et al. discovered that parsing algorithms can be defined over arbitrary semirings. These semirings then determine the output of a parsing algorithm. The smallest such semiring is the boolean semiring. This semiring provides an output, 0 or 1, that describes whether or not the grammar recognizes the output. In contrast, the largest such semiring is the free semiring over the grammar productions. In effect, the output of a parsing algorithm over this semiring will give a ``complete parse forest,'' a description of all possible parses of the provided string. Goodman's thesis~\cite{goodmanthesis} describes an implementation of Earley parsing over arbitrary semirings that we utilized, so we merely needed to find the correct semiring for our algorithm.

Provenance semirings are a well-studied representation for formalizing provenance information of database queries. One such provenance semiring is the ``conditional tables'' semiring. This semiring effectively represents boolean formulas over parse rules. Parsing with the conditional tables semiring outputs a boolean formula that describes which rules must be included for a valid parse to be extracted -- exactly what we're looking for!

\subsection{SMT Solver Integration}

We use Z3~\cite{z3} to solve our MaxSMT calls. Of course, Z3 doesn't permit using grammar productions as boolean variables. We build a mapping between grammar productions and Z3 boolean variables. If Z3 finds a satisfying assignment to variables, it builds a grammar from the productions associated with the variables assigned to true.
}

\OMIT{\afm{There was a lot of great stuff in this section, but I think we can just refer to prior work. This being said, I think right now I have too little intuition provided, I just cite the two relevant things. So I want to keep this stuff in the document, as we can use already written stuff to describe}

In this section, we will go into further detail about how we implemented a language for specifying a 
grammar induction problem as well as how we implemented each of the three grammar induction algorithms.

\subsection{Grammar Induction Algorithm Design}
\label{induction-algo-implementation-design}
Now that we have described the induction problem language in depth, we move to algorithm implementation.
We will go over the shared and/or branched implementation first before diving into
implementation details specific to each algorithm. All three algorithms rely on an SMT solver;
in particular, we use the Z3-SMT solver which has a nice OCaml interface at 
\url{https://github.com/termite-analyser/z3overlay}. As for provenances and rule collections, both are 
implemented as sets of (Name: string, Index: int) and (Index: int) respectively. The former could have been 
a true alias of rule collections, but it didn't make much difference in runtime.
Moving on, the Prosynth-inspired and SATLite algorithm both use a Delta-debugger. 
Due to its prevalence and simplicity, there are plenty
of algorithm outline resources to aid in implementation \cite{DeltaDebugging}. Finally, the body of the algorithms are not very 
exciting -- after the dependencies are loaded, they are used as black-boxes in simple control statements.
There are no new ideas in our work extending these well-established tools. \par
The novelty of our work comes when implementing the context-free-grammar related dependencies. Between the three
algorithms, one must implement  
\begin{enumerate}
    \item A Recognizer \rec: $\mathcal{M} \times E \to bool$
    \item A Provenance Finder \prov: $\mathcal{M} \times E \to prov$
    \item A Provenance Oracle Constructor $\textbf{OC}_{\mathcal{M}}: E \to \mathcal{PO}$
    \item A Provenance Condenser $\textbf{PC}_{\mathcal{M}}: E \to \Phi$
\end{enumerate}

General context free parsing is well studied and we decided to base all of the objects above on the
Earley recognition algorithm. A nice explanation of the algorithm and its implementation can be 
found on Loup Vaillant's blog page \cite{EarleyImplementation}. \par 
Note that the recognizer \textbf{R} does \emph{not} have anything to do with provenances. It is not difficult 
or computationally expensive to adapt the recognizer to keep track of a single provenance -- a single 
provenance can be recorded for each Earley item in the parse table and the Early algorithm will terminate in the 
same number of steps as before. Provenances are implemented as sets and joining them is not too expensive. 
Our final implementation uses a this type of provenance puller for the Prosynth-inspired algorithm.
\subsubsection{Full Provenance Tracking}
Looking at the prerequisites for the SATLite and SATHeavy algorithms, it's not even clear at first that the provenance 
oracle or condenser can be constructed in a non-combinatorial manner. Fortunately, the problem of efficiently
tracking all provenances used to derive an example has been studied and documented \cite{EarleyDerivationTree}. At a high
level, Earley recognition can be augmented with any semi-ring that stores information about the recognition paths. For example,
when augmented with the semiring $<\mathcal{N}, +, \times, 0, 1>$, the final Earley item contains the number of 
distinct ordered rule-applications which can be used to parse a string. \par 
For our purposes, we augment the Earley algorithm with the semiring $<\mathcal{M}, \cup, \cdot , \emptyset, \{\{\}\}>$
where $\cdot$ stands for concatenation. The idea here is relatively natural at a high level: when there are two ways 
to reach the same parse table item, we note that with the union $\cup$; when two items are needed to construct a third item, 
we note that with a concatenation $\cdot$. The previous statement does not account for nullable grammar symbols but these are covered
in detail in the parsing reference \cite{EarleyDerivationTree}. \par
Ultimately, the values output by the augmented Earley recognition algorithm are expressed with 
type in Figure \ref{fig:provenance-tracker}.\par 
\input{figures/prov-tracker-signature}
The provenance tracker type has very simple but powerful primitives. We will showcase their uses now. 
\subsubsection{Adapting the Provenance Tracker}
The augmented Earley algorithm translates beautifully into the objects needed for the SATLite and SATHeavy
algorithms. \par
To implement the SATLite algorithm, we need a provenance oracle constructor $\textbf{OC}_{\mathcal{M}}: 
E \to \mathcal{PO}$. Under the hood, a provenance oracle is just a provenance tracker(!) so a provenance
oracle constructor is just the augmented Earley algorithm itself! We just need to show that a provenance
tracker inhabits the provenance oracle type. Recalling the definition in \ref{subsec:satlite-alg}, a 
provenance oracle must support the Recognize and GetProv functions -- we will show that a provenance
tracker can do so efficiently.

Suppose we have constructed a provenance tracker $T_e$ for metagrammar $\mathcal{M}$ and example $e$. To
implement Recognize, we need to check if grammar $G \in \mathcal{M}$ ``covers'' a path in $T_e$. We can unpack 
the semantically-suggestive primitive names and realize that for $G \in \mathcal{M}$ and tracker $T$, if T matches
\begin{enumerate}
    \item Or subtrackers -> $G$ covers $T$ iff it covers at least one of the subtrackers
    \item And subtrackers -> $G$ covers $T$ iff it covers all of the subtrackers
    \item Rule (ruleNum, subtracker) - $G$ covers $T$ iff it contains ruleNum and covers subtracker.
\end{enumerate}
Note that rule 1 and 2 above handle the trivial cases Or/And [] correctly. Or [] denotes
that there is no parse and so no tracker is covered; And [] denotes that no rules are 
needed and therefore all trackers are trivially covered. And of course, we have that $G$ covers
$T_e$ iff $G$ recognizes $e$. \par
Similarly, to implement GetProv, one notes that for $G \in \mathcal{M}$ and general $T$ matching
\begin{enumerate}
    \item Or subtrackers - $G$'s provenance is \emph{one} of the provenances from the
    subtrackers. If none exists, then $G$'s provenance is None.
    \item And subtrackers - $G$'s provenance is the union of all provenances
    from the subtrackers. The union of a provenance with None is None. 
    \item Rule (ruleNum, subtracker) - if ruleNum is in $G$'s rules, then 
    $G$'s provenance is ruleNum added to its provenance for the subtracker. 
    Otherwise, it is None. 
\end{enumerate}
Note from item 3 above, we will always pull a provenance of $e$ which uses rules only in $G$. Moreover,
we only need to keep track of a single provenance so the number of objects we need to keep track of 
is limited. Then, as desired, we have that GetProv with $T_e$ on $G$ returns a correct provenance if 
$G$ recognizes e and None otherwise.\par
Together, these ideas can be used to efficiently implement a Recognize and GetProv function for a tracker,
meaning that it makes an efficient oracle for the SATLite algorithm.

Moving onto the SATHeavy algorithm, we need to specify how to create a provenance condenser 
$\textbf{PC}_{\mathcal{M}}: E \to \Phi$ from the provenance tracker. A bad implementation 
of condenser keeps a set of SAT formulas corresponding to each provenance and $Ors$ them together at the end -- 
this is terrible because it does not use the structure of the provenance tracker at all. 
Instead, for metagrammar $\mathcal{M}$ and example $e$, we construct a provenance tracker $T_e$ for the specified 
example. We then recursively define a SAT formula by noting that for general tracker $T$ matching 
\begin{enumerate}
    \item Or subtrackers - $T$ is parsable iff at least one of the subtracker SAT formulas is satisfied
    \item And subtrackers - $T$ is parsable iff all of the subtracker SAT formulas is satisfied
    \item Rule (ruleNum, subtracker) - $T$ is parsable iff ruleNum is included and the subtracker SAT formula is satisfied 
\end{enumerate}
The note above allows us to condense a large SAT formula $\phi_e$ from a provenance tracker $T_e$ such that
grammar $G$ parses $e$ iff its included candidate rules are a satisfying assignment of $\phi_e$. Moreover, notice
that we don't run into a combinatorial explosion problem if we do a direct translation of $And$ and $Or$ into SAT
formulas: the number of literals is exactly the number of trackers in $T_e$ of type Rule. 

So, we are able to construct a provenance oracle and provenance condenser from provenance trackers! To our knowledge, 
the semiring augmentation of Earley parsing in the past was theoretically interesting, but did not have any convincing
application. Here, we show that the novel syntax-guided grammar induction problem is the perfect playground for this machinery. }


\section{Experimental Results} 
\label{sec:experiments}

To illustrate \Sys’s practicality, we evaluate the following properties:
\begin{itemize}
    \item \textbf{Expressiveness} Can \Sys specify real-world grammatical domains?
    \item \textbf{Time Efficiency} How much time does \Sys take to induce a grammar?
    \item \textbf{Sample Efficiency} How many examples does \Sys take to induce a specific grammar?
    \item \textbf{Comparisons} How does \Sys compare to prior work?
\end{itemize}

\paragraph{Experimental setup}
All experiments in this section were performed on a 2.5 GHz Intel Core i7 processor with 16 GB
of 1600 MHz DDR3 RAM running macOS Catalina. We ran each benchmark 10 times,
with a timeout of 60 seconds, and report the average time. If any of the 10 runs
times out then we consider the benchmark as a whole to have timed out.

\subsection{Expressiveness}
\label{subsec:expressiveness}

\begin{figure}%
  \small
  \centering%
  \begin{tabularx}{\linewidth}{|r|c|c|c|X|}
    \hline
    \multicolumn{5}{|c|}{\textbf{SGI Benchmark Suite}}\\
    \hline Name & Nonterms. & Prods. & AST Nodes & Description \\
    \hline
    States & 6 & 9 & 371 & US State identifiers.  Permits acronyms, full names, and abbreviations.\\
    \hline 
    Phone \#s & 11 & 44 & 260 & Phone numbers.  Permits local and international phone numbers. \\
    \hline
    Times & 12 & 21 & 204 & Time of day in a variety of formats.\\
    \hline 
    Floats & 12 & 18 & 113 & Floating-point numbers.  Includes grammars for different scientific notations and standard decimal form.\\
    \hline
    Emails & 15 & 105 & 493 & Email addresses. Includes grammars that accept emails from specific or arbitrary domains.\\
    \hline 
    Names & 11 & 32 & 142 & Human identifiers. Includes grammars specifying salutations, post-nominal titles, and acronyms.\\
    \hline
Streets & 14 & 40 & 167 & US street identifiers. Includes grammars that demand specific suffixes and directions.\\
    \hline 
    Dates & 18 & 37 & 314 & Calendar dates. Includes month-first, day-first, and year-first formats.\\
    \hline
    Addresses & 28 & 58 & 597 & US street addresses. Uses the States and Streets metagrammars to identify those portions of the address.\\
    \hline
    XML & 21 & 91 & 523 & XML Files. Permits 10 classes of XML elements. It discovers the identifiers for element classes and the recursive schemes. In effect, it imputes the structural component of a schema definition. \\
    \hline
    SQL & 25 & 31 & 357 & SQL: SQL SELECT statements. Supports nested joins and extra keyword clauses like WHERE and LIMIT. \\
    \hline
    IdealCSV & 38 & 61 & 506 & CSV: Idealized version of CSV. Based on RFC 4180~\cite{csvrfc}, but also admits common kinds of separators including tab and semi-colon. Automatically infers cell type.
     \\\hline
  \end{tabularx}
  \caption{Information on metagrammars for the SGI benchmark suite. For each, we include the number of nonterminal definitions, the total number of productions, the total number of AST nodes. }%
  \label{fig:benchmarkinfo}
\end{figure}

We began our evaluation by identifying 12 grammatical domains to use in our experiments. We intentionally chose domains that range from simple (dates, times) to complex (XML, CSV). For each domain, we manually searched the internet to identify various formats that were used in practice. Our sources included actual data files, as well as documents containing descriptions of formats. 

We then \emph{manually} constructed a corpus of example files adhering to each format. As a result, our examples are somewhat artificial---we wrote them ourselves---but nonetheless represent a broad range of real-world formats.

Once we had conducted our survey of real-world data formats, we used that experience to write metagrammars for each domain. Figure~\ref{fig:benchmarkinfo} (in the body of the paper) describes these domains briefly, as well as the complexity of each. Qualitatively, we did not find these grammars particularly hard to write---the domains do not require many productions and most of the insight was in what preferences and constraints were necessary. On occasion, our initial domain definitions were erroneous, in the sense that they did not return the ``best" grammar for certain examples. However, we found such errors were relatively easy to fix by adjusting preferences. Since we were able to successfully write metagrammars for each domain, we conclude that \Sys \emph{is} capable of representing real-world domains.

\subsection{Time Efficiency}
\label{subsec:time-efficiency}

\begin{figure}
    \includegraphics[scale=.84]{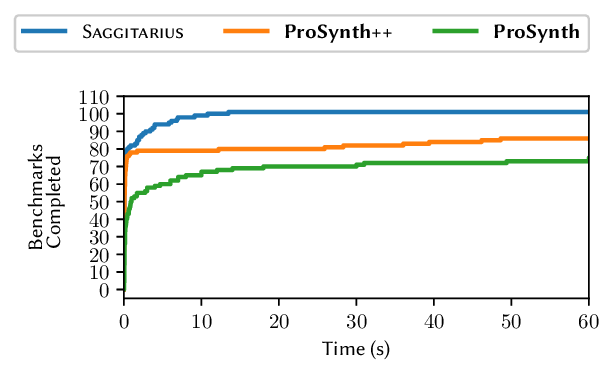}
    \caption{Number of benchmarks that terminate in a given time in different modes. The \textbf{ProSynth} line represents using the ProSynth algorithm, originally built for learning logic programs, for grammar induction. The \textbf{ProSynth++} line represents using the ProSynth algorithm, but with an optimization that minimizes duplicate parsing calls.}
    \label{fig:times}
\end{figure}

\paragraph*{Benchmarks}
To evaluate the efficiency of \Sys, we developed a benchmark suite of 10
induction tasks, each of which ask \Sys to induce a grammar given a certain number of positive examples (PEs) and negative examples (NEs). Each benchmark varies the number of PEs and NEs, as well as whether or not a suitable grammar exists in the metagrammar. Specifically, our benchmarks are to induce a grammar given:
\begin{multicols}{2}
\begin{enumerate}
\item 1 PE.
\item 1 NE.
\item 1 PE and 1 NE.
\item 10 PEs.
\item 5 PEs and 5 NEs.
\item 1 PE and 20 NEs.
\item 20 PEs and 20 NEs.
\item 1 PE and 1 NE, when no appropriate grammar exists.
\item 10 PEs and 10 NEs, when no appropriate grammar exists.
\item 20 PEs and 1 NE, when no appropriate grammar exists.
\end{enumerate}
\end{multicols}

In this experiment, we are are not asking \Sys to return a particular grammar; rather, we are evaluating how long it takes to return \emph{some} grammar (or "no grammar exists", for the latter three tasks). We evaluate \Sys's ability to return specific grammars in the next section. 

For each task and each grammatical domain in Figure \ref{fig:benchmarkinfo} except SQL (due to time constraints), we selected a set of examples from our corpus to create a total of 110 benchmarks. We evaluated \Sys by running it on these 110 benchmarks using three different algorithms, meant to simulate the performance of related tools. We recorded how long each benchmark took to finish, and summarize the results in Figure \ref{fig:times}. This figure is also used later in Section~\ref{subsec:comparisons}. We find that \Sys solves $102$ of our $110$ benchmarks in under a minute, and typically does so in under 10 seconds.

\begin{figure}
    \includegraphics[scale=.84]{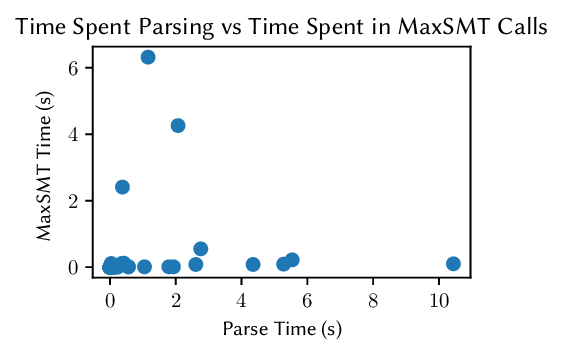}
    \caption{The time spent during the parsing phase of the benchmark suite compared to the time spent during the MaxSMT solving phase of the benchmark suite. Aside from a few outliers, the majority of the time is spent during parsing.}
    \label{fig:time-breakdown}
\end{figure}

Figure~\ref{fig:time-breakdown} shows how much time was spent for each benchmark in the two main phases of our algorithm: parsing and performing MaxSMT calls. In a majority of these benchmarks, parsing dominates the runtime, and all but one of the failing benchmarks hang during parsing. \Sys uses an Earley parser~\cite{earley,EarleyImplementation} modified to operate over semirings, though 
other general context-free parsing algorithms, such as a GLR algorithm~\cite{GLR}, would also have been suitable.
We believe that further developments in semiring parsing could substantially increase the efficiency of \Sys on many benchmarks.
For example, recent work has found that semirings with rich sets of equivalences (like the idempotence of $+$ and $\times$ in the conditional table semiring) provide additional opportunities for memoization in the parsing algorithm~\cite{memoizationparser}. Furthermore, in some simpler domains like Dates and Phone Numbers, every grammar in the domain falls into smaller complexity classes than merely context-free, like LL(*). By integrating the faster parsing algorithms known to work on these simpler domains, induction could likely happen dramatically faster.

However, in one domain MaxSMT, rather than Early parsing, dominates the runtime -- XML. For each example in the XML domain, a majority of the time is spent during the MaxSMT call, and one XML benchmark fails to complete the MaxSMT call in the allotted timeframe. This is because the MaxSMT algorithm has a difficult time proving a given model is optimal in the XML domain. If we replace the MaxSMT call by a simple SMT call (equivalent to removing preferences from our source file), these benchmarks finish near instantly. We perform a more involved case study on the XML domain in\ifappendix{} Appendix~\ref{sec:case-studies}. \else{}
the full version of the paper.\fi We believe that developments into more efficient MaxSMT algorithms could increase the efficiency of \Sys without sacrificing optimality.

\subsection{Sample Efficiency} 

\begin{figure}
    \begin{subfigure}[b]{0.60\textwidth}
    \includegraphics[scale=.8]{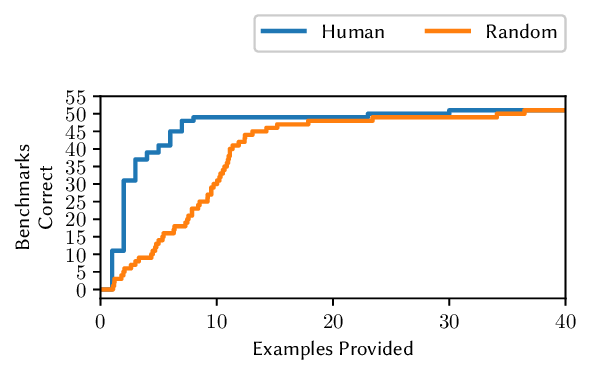}
    \caption{\# Correct benchmarks vs. \# examples used.}
    \label{fig:example-5.3-graph}
    \end{subfigure}
    \begin{subfigure}[b]{0.346\textwidth}
    \begin{tabularx}{\textwidth}{|c|c|c|}
        \hline
        Name & Human & Random \\\hline
        States & 2 & 6.5 \\\hline
        Phone \#s & 3 & 10.3 \\\hline
        Times & 4.5 & 12.4 \\\hline
        Floats & 1.8 & 6.2 \\\hline
        Emails & 2 & 5.9 \\\hline
        Names & 6.5 & 12.5 \\\hline
        Dates & 4 & 10.6 \\\hline
        Addresses & 8 & 13.9 \\\hline
        IdealCSV & 1 & 3 \\\hline
    \end{tabularx}
    \caption{Average \# examples needed per domain}
    \label{fig:example-5.3-table}
    \end{subfigure}
    \hfill\\\hfill\\
    \caption{Sample efficiency of \Sys. Human examples were chosen to be as helpful as possible, in the human's judgement.}
    \label{fig:example-5.3}
\end{figure}

One of the advantages of \Sys is its ability to successfully induce useful grammars from very few example. To demonstrate this, we designed an experiment to measure how many examples it takes to induce a particular grammar, or correctly claim that no suitable grammar exists. 

Our setup is as follows: for each terminating benchmark that involves at least 5 examples, run \Sys on a single example. If the output is not the desired one, select another example, and run \Sys on both. Continue adding examples until the output matches our expectation.

Figure \ref{fig:example-5.3} summarizes our results using two different selection strategies. In the first, we have a human operator select examples they think will be the most useful. In the second, we select each example at random. We can clearly see from  Figure \ref{fig:example-5.3-graph} that nearly all benchmarks can be completed with fewer than 10 well-chosen examples. Even when examples are chosen at random, \Sys almost never needs the entire set of examples, settling on the desired grammar with between 10 and 20 inputs.

Figure \ref{fig:example-5.3-table} breaks down the data according to grammatical domain. Intuitively, grammatical domains where the grammars have a lot of ``overlap'' are generally harder to specify. For example, a state can be represented as with its full name, the two-letter abbreviation, a historic acronym, or some combination of the three. From a single example, it will not be clear if only one representation is acceptable, or if multiple types are allowed. 

When domains are combined, the ambiguity of the components compounds. For example, addresses must determine not only which state representations are allowed, but whether street names are written fully (e.g. ``Street''), abbreviated (``St.''), or omitted entirely. The decisions for states, streets, and other components such as names must be made separately, and it is difficult to constrain all of them in a single example,  particularly when your desired address format allows several possibilities for each.


We omitted the XML domain from this experiment 
because the XML meta-grammar can generate several 
syntactically different, but semantically equivalent
grammars.  As a consequence, determining whether the
tool generates the correct grammar automatically is
much more difficult than just running \texttt{diff}. 
Rather than attempting to implement equivalence checking, we opted to skip the XML domain for this experiment.

\subsection{Comparisons}
\label{subsec:comparisons}

\Sys is relatively unique in the problem it solves; there are few existing tools which do the same thing.
The closest analogue we know of is the inductive logic programming engine \emph{ProSynth}. Since parsing via context-free grammar is a refinement of logic
programming, the ProSynth algorithm may be effective. Unfortunately, the process of transforming a \Sys program into a ProSynth program proved impractical. Instead, we compare \emph{strategies} by implementing two additional induction algorithms for \Sys, which simulate ProSynth's process. We can then compare the performance of these algorithms directly.

The three stratgies we evaluate are:
\begin{itemize}
\item \SysMode{}: The algorithm described in \S\ref{sec:grammar-induction-algorithms}.
\item \ProsynthMode{}: In this mode, \Sys iteratively guesses subsets the candidate rules, 
and checks to see whether the subset parses the positive but not the negative examples.
If it fails, it uses a counter-example-guided
algorithm to generate a new set of candidate rules and repeats the process. 
If it is able to guess an appropriate set early, it may avoid the heavy cost of parsing all data
with a large ambiguous grammar.
\item \ProsynthPPMode{}:  In this mode, \Sys parses all the example data, and constructs a full provenance tree as in \SysMode{}.
  However, instead of building a single large MaxSMT formula, it considers
  candidate rule sets in a counter-example-guided loop, as in \ProsynthMode{}. 
  This mode may avoid constructing a large SAT formula, at the cost of making many SAT calls.
\end{itemize}

\paragraph*{Results}
We can see from Figure \ref{fig:times} that our primary algorithm (\SysMode) completed 101 out of 110 benchmarks
in under 20 seconds. In 6 of the 9 instances the system timed out, it did so because
no grammar matched the example data (by experiment design) and \Sys looped, forever trying larger and larger values of its existential variables.
A useful improvement to the system would be to implement heuristics that can detect when no grammar in an infinite metagrammar will satisfy a data set. The remaining 3 of the 110 experiments ran long due to particularly slow MaxSMT calls. Removing all preferences drastically speeds up these tasks, allowing them to complete in under a minute (though the returned grammars are not the desired ones).

In comparison to the ProSynth modes, we found that \SysMode solved every benchmark but one faster than both
\ProsynthMode{} and \ProsynthPPMode{}. The exception was one of the benchmarks on which
\SysMode timed out on an SMT call.
Although the multiple smaller SMT calls of \ProsynthPPMode usually take longer in aggregate, in this one
case they avoided generating a particular challenging SMT formula.

More broadly, the cost of all of the algorithms grows substantially with the number of
rules in the metagrammar.  The ProSynth-inspired algorithms, in particular, suffer significantly
as rule sets increase because it takes them more iterations to build up a sufficiently constrained SAT call
to find a solution. \Sys is also negatively impacted by having more rules, but in a different (and less substantial) way: 
more rules results in slower parsing, which is the typical bottleneck for \Sys (complex MaxSMT calls are most costly, but much rarer).

Across all the algorithms, we found that the \emph{number} of examples did not affect
system performance substantially.  However, the \emph{length} of an example can have a significant
impact when the metagrammar is highly ambiguous. Earley parsing with large numbers of ambiguous
rules is a costly enterprise that can grow cubically with the length of the input in the worst case.
This parsing cost tends to dominate as examples grow in length.

\subsection{Case Study: CSV Dialect Detection}
\label{subsec:case-study-dialect}

As part of our comparisons, we describe our case study which compares \Sys to a domain-specific tool on an individual grammatical domain: CSV. We have performed additional case studies on the XML and SQL grammatical domains, but have elided them for space, the interested reader can find them in\ifappendix{} Appendix~\ref{sec:case-studies}. \else{}
the full version of the paper. \fi
RFC 4180 provides the standardized definition of CSV, but there are many real-world data files
that do not conform to the standard.  Such data files may contain irregular ``table-like'' data in which 
rows have differing numbers of columns, cells are malformed (containing dangling delimiters), and 
the entire table is surrounded by unstructured text.  

To handle such messy CSV-like data, 
Van den Burg \emph{et al.}~\cite{csv-wrangling} designed CSV Wrangler, a tool for inferring the 
\emph{dialect} of messy, possibly erroneous CSV documents  
\footnote{A CSV dialect is a triple of a cell delimiter, a quote character (for strings), and an escape character.}.
\OMIT{
The delimiter separates cells; the quote character surrounds strings, allowing delimiters to appear within a cell, and the escape character admits nested quotes. 
For example, 
the row below should be parsed with a comma delimiter, a quotation mark as the quote symbol and
a backslash as the escape character. 
\begin{verbatim}
    0, "\"How's the weather?\" she asked.", "1"
\end{verbatim}
}
Unfortunately, CSV files are inherently ambiguous. To handle this ambiguity,
Van den Burg \emph{et al.} use a custom algorithm that
scores potential dialects based on how consistent the resulting rows are (the \emph{pattern score}) and 
how well-typed the cells are (the \emph{type score}) ~\cite{csv-wrangling}, attempting to mimic the
human process of recognizing the data in a messy CSV file. 

\OMIT{
Python's CSV sniffer~\cite{python-sniffer} is a similar sort of dialect detector.
Dialect detection is difficult because many real CSV dialects are ambiguous. For example, a CSV in our case study contains many lines of the form
\begin{center}
\begin{BVerbatim}
(tab)(tab)   C1(tab)   C2,C3(tab)  C4,C5(tab)  C6
\end{BVerbatim}
\end{center}
The resulting dialect is ambiguous as it could be part of a well-formed CSV with either the tab or comma delimiter. Under RFC guidelines, the CSV should be parsed with the comma delimiter, but
the CSV Wrangler tool and the human labeler used in experiments to evaluate CSV Wrangler chose the tab delimiter~\cite{csv-wrangling-github}.
}

To experiment with CSV dialect detection using \Sys, we developed two additional CSV metagrammars, named \emph{Strict} and \emph{Lax}. The former is more restrictive, hewing closer to the RFC, while the latter attempts to mimic the same human priors as Van den Burg's tool.\footnote{
Specifically, Strict requires that (1) nested quote characters do not appear in any cell
and (2) delimiters (commas, semicolons, tabs and vertical bars) internal to a cell must be quoted. Lax contains neither restriction.} As a result, Lax is much costlier to execute. Fortunately, because Lax is more lenient than Strict, the two processes can be pipelined --- we try Strict first, and fall back to Lax only if it fails. We call this pipelined system StrictLax.

\OMIT{
Lax is considerably more costly to execute than Strict. Lax is highly ambiguous
as it admits dangling quotations 
and ill-formed cells, and uses a preference system to sort through the multitude of possible parse trees.}

\paragraph*{Benchmarks}
We measured the performance of \Sys on 100 ASCII-128 CSVs with human labels drawn from
Van den Burg \emph{et al.}'s GitHub repo~\cite{csv-wrangling-github}. More than half of the benchmark suite (59/100) does not obey the CSV RFC.  In such situations, it is difficult
for humans to unambiguously identify the dialect of files, as we discuss below.

\paragraph*{Experiments}
We attempted dialect detection using 5 different tools: (1) RFC (the \Sys specification of the
CSV RFC), (2) Strict, (3) StrictLax (Strict followed by Lax), (4) CSV Wrangler
(Van den Burg's tool), and (5) the Python Sniffer~\cite{python-sniffer}. 
We found that Earley parsing was the primary bottleneck for \Sys system on large CSV benchmarks, so 
we truncated the CSV files to 20 lines prior to
processing them with Strict or RFC, and to 5 lines prior to processing them with Lax\footnote{While this may seem extreme, one would expect \emph{all} lines of a CSV file to obey the same format. Hence a very small number of lines is still likely sufficient to infer the dialect.}.

Figure~\ref{fig:csv-exp} presents the results from our tool as well as corresponding results
we retrieved  from 
Van den Burg \emph{et al.}'s GitHub repository \cite{csv-wrangling-github} for CSV Wrangler and Python Sniffer.
We find that \Sys performs only slightly worse than dedicated CSV sniffers -- in particular, the number of incorrectly classified files is nearly identical across all the tools. Furthermore, we believe that all but 2 of our misclassifications represent reasonable behavior by \Sys, even if the human classifier chose a different dialect. The misclassifications are described in Figure \ref{fig:misclassification}.

Unlike dedicated sniffers, \Sys must contend with the possibility of timing out, either due to infinite metagrammars or simply slow analysis. Although we see \Sys timing out on a noticeable fraction of examples, profiling indicates that much of this time is spent performing Earley parsing. It is likely that we could significantly improve our performance by adopting an industrial-strength, optimized parser in place of our current, unoptimized one.

\begin{figure}
\begin{center}
\begin{tabular}{|c|c|c|c|c|}
    \hline
    Detector & Yes & No & No Dialect & Timeout \\ \hline
    RFC & 41 & 14 & 44 & 0 \\
    Strict & 67 & 12 & 8 &  13 \\
    StrictLax  & 70 & 14 & 0 & 16  \\ 
    CSV Wrangler & 86 & 13 & 1 & N/A \\ 
    Python Sniffer & 81 & 14 & 5 & N/A \\ \hline
\end{tabular}
\end{center}
\caption{CSV Analysis. We report the number of files on which the tool: aligns with
the human label (\textbf{Yes}); does not align with the
human label (\textbf{No)}; reports that no dialect exists (\textbf{No Dialect}); or times out (\textbf{Timeout}).}
\label{fig:csv-exp}
\end{figure}

\begin{figure}
\begin{center}
\begin{tabular}{|c|c|c|}
    \hline
    Id & \Sys behavior &  Frequency \\ \hline
    1** & \shortstack{CSV preferences allowed algorithm \\ to return no delimiter character} & 4 \\ \hline
    2* &\shortstack{chose different delimiter than \\ hand-label but file was ambiguous.} &  2 \\ \hline
    3* &\shortstack{did not find any characters because \\ of leading metadata} & 2 \\ \hline
    4 &\shortstack{did not find quote or escape character \\ because of truncation } & 2 \\ \hline
    5* &\shortstack{did not accept space \\ as a valid delimiter} & 1 \\ \hline
    6** &\shortstack{did not identify quote in \\ ill-quoted file} & 1 \\ \hline
    7** &\shortstack{unquoted ``,'' character lead to \\ false ``,'' delimiter} & 1 \\ \hline
    8* &\shortstack{did not find improperly \\ used escape character} & 1 \\ \hline
\end{tabular}
\end{center}
\caption{Reasons for StrictLax misclassification.  In our analysis, cases marked * are 
 inherently
ambiguous or mislabelled by the human. Cases marked ** also represent reasonable behavior by
our system in our judgement.}
\label{fig:misclassification}
\end{figure}

\OMIT{
\paragraph{Misclassifications}
We investigated each of the circumstances in which \Sys varied from human-labelled data, and found 8 different scenarios; these are listed in
Figure~\ref{fig:misclassification}.
We believe that in scenarios 2, 3, and 8 (5 of the 14 misclassifications),
our tool is correct, or at least not wrong. Furthermore, after hand inspection of the files,
we believe that the tool returns
reasonable answers in situations 1, 6 and 7. In scenario 5, we even have disagreements among ourselves -- one author believes spaces should not be
allowed as delimiters in CSV, so this case certainly is not clear either.
The remaining two failures (scenario 4) are caused by truncation. These are clear failures of \Sys, but could be addressed by more efficient parsing.
}

\paragraph*{Takeaways} First, we have shown that
it is possible to develop more than one specification for an ambiguous grammatical domain.  Multiple generated tools
can then be arranged in a pipeline and take advantage of
time/accuracy tradeoffs.  Second, real-world data is often
messy.  Hand-tuned, custom tools such as
Python's CSV sniffer or the CSV Wrangler can make
mistakes; even humans often disagree with each other.  
Nevertheless, \Sys performs comparably to these
custom-built tools.  The CSV sniffer and the CSV Wrangler
align with outside human labeller a few more times than
our tool, but the authors of this paper
actually disagree with the outside labels in almost all such cases. 
Perhaps the labeller and tool authors know of some
criteria for disambiguating CSV data that we do not; if so,
we can likely add this criteria
to our own specifications.
When it comes to the unambiguous formats where humans agree, \Sys is very effective.

\section{Related Work}%
\label{sec:related-work}

\paragraph*{Grammar induction} Grammar induction 
traces back to at least the 60s when Gold~\cite{gold:inference} began studying
models for language learning and their properties.  Later, Angluin~\cite{angluin:regular-sets-complexity,angluin:regular-sets} developed her
famous L$^*$ algorithm for learning regular languages.  As mentioned earlier, however,
such algorithms often require large numbers of examples even for simple regular expressions. More recently, FlashProfile~\cite{flashprofile} has shown that regular-expression-like \emph{patterns} can be learned from positive examples, by first clustering by syntactic similarity, and then inducing programs for given clusters. 

Inference of context-free grammars is
considerably more difficult, and results are limited. It has been 
tackled, for instance, by Stolcke and Omohundro~\cite{Stolcke:grammar-inference}, who use probabilistic techniques to infer grammars. Fisher \emph{et al.}~\cite{pads-inference} explored
inference of grammars for ``ad hoc'' data, such as system logs, in the context of the PADS project~\cite{pads}. Lee~\cite{alpharegexp} developed more efficient search
strategies for regular languages in the context of a tool for teaching automata theory.  Both these latter tools tackled 
restricted kinds of grammars, however; scaling to complex formats using few examples remains a challenge. 
A more recent approach can be seen in the \Glade~\cite{glade} tool. Similarly to L$^*$, \Glade{} uses an active learning algorithm, but generalizes to full context-free grammars, rather than merely regular expressions, while requiring relatively fewer membership queries.

We believe that the key contributions of this paper are largely orthogonal to these advances. 
In particular, we introduce the idea of using \emph{grammatical domains}, specified via metagrammars, 
to restrict the set of grammars under consideration during the induction process. Doing so has the potential to improve the performance of almost any grammar induction algorithm. 

Grammatical inference becomes more tractable when one can introduce bias or constraints---metagrammars are one way to introduce such bias but there are others.  For instance, Chen \emph{et al.}~\cite{chen2020multimodal}
use a combination of examples and natural language to speed inference of 
a constrained set of regular expressions.  Internally, their system 
generates an ``h-sketch'' as an intermediate result.
These h-sketches are partially-defined regular expressions that may include holes for 
unknown regular expressions.  Such h-sketches play a similar role to our 
metagrammars: they denote sets of possible regular expressions and constrain 
the search space for grammatical inference.  However, our language is an
extension of YACC and is designed for humans, rather than being an
intermediate language. Furthermore, our metagrammars may be 
reused, like libraries, across data sets. In constrast, each h-sketch is generated and used only once inside a compiler pipeline.  

Related to the notion of grammatical inference is that of expression \emph{repair}. \RFixer~\cite{rfixer} uses positive and negative examples to fix erroneous regular expressions. Both \RFixer{} and Saggitarius use similar algorithms to ensure positive examples are in the generated language, and negative examples are not. Both of these tools encode these constraints as MaxSMT formulas to ensure the generated grammars are optimal. Because RFixer does not have a metagrammar to orient the search, their constraints can only help find character sets that distinguish between the grammars. Saggitarius permits any constraints that expressible in propositional logic, and the constraints can be over arbitrary productions, not merely character sets. One could see the RFixer algorithm as an instance of our algorithm, where the meta-grammar constrains sets of allowed characters.

\paragraph*{Syntax-guided Program Synthesis} Our work was inspired by the progress on
syntax-guided program synthesis over the past decade or so~\cite{sketch-original,sketch-asplos,sygus,sketch-summary}.  Much of that
work has focused on data transformations, including 
spreadsheet manipulation~\cite{flashfill,dace,flashrelate},
string transformations~\cite{synth-symm-lenses,optician,blaze}, and information extraction~\cite{flashextract,data-extraction}.  Such problems have much in common with our work, but they
have typically been set up as searches over a space of program transformation operations
rather than searches over collections of context-free grammar rules.  Particularly
inspiring for our work was the development of FlashMeta~\cite{flashmeta} and Prose~\cite{prose}. These systems are ``meta'' program synthesis engines---they help
engineers design program synthesis tools for different domain-specific languages.  Similarly, \Sys is a
``meta'' framework for syntax-guided grammar induction, helping users perform grammar induction in domain-specific contexts.  Of course, \Sys, FlashMeta and Prose differ greatly when it comes to specifics 
of their language/system designs and the underlying search 
algorithms implemented.

\paragraph*{Logic Program Synthesis} We were also inspired by work on Inductive Logic Programming~\cite{ilpbook}, and Logic Program Synthesis~\cite{Prosynth,difflog}. Parsing with context-free grammars is a special case of logic programming so it was natural to investigate whether inductive logic programming algorithms would work well here. ProSynth~\cite{Prosynth} is a state-of-the-art algorithm in this field so we experimented with it as a tool for
grammatical inference. However, we found our custom algorithm almost always
outperformed ProSynth on grammatical inference tasks.

\section{Conclusion}%
\label{sec:conclusion}

\emph{Grammatical domains} are sets of related grammars.  Such domains appear naturally whenever a common
datatype like a date or a phone number has multiple textual representations. They also appear frequently when
data sets are communicated via ASCII text files, as is the case for CSV files.  In this paper, we
introduce the concept of grammatical domains, provide a variety of examples of such domains in the wild, 
and 
design a language, called \Sys, for defining grammatical domains through the specification of metagrammars.

\Sys includes features for defining sets of candidate productions, for
constraining the conditions under which candidate productions may and may not appear, and for ranking the generated grammars.
We illustrate the use of \Sys on a variety of examples and develop a grammar induction algorithm
for the system that uses semiring parsing to generate MaxSMT formulas that can be solved via an
off-the-shelf theorem prover.  

In the future, we look forward to developing
a complete parser generator system using the ideas
developed in \Sys.  One way forward would be to
add semantic actions to \Sys, making it much more like
YACC.
Another direction would involve integrating \Sys into
a parser combinator library such as Parsec~\cite{parsec}.



\begin{acks}
\end{acks}

\bibliography{paper}

\clearpage
\ifappendix
\appendix
\section{Full Algorithm and Proofs of Theorems}%
\label{sec:proofs}

The Algorithm is described fully in Algorithm~\ref{alg:grammarinduction-full}

\begin{algorithm}[!b]
\begin{algorithmic}[1]
\Procedure{AllProductions}{$M,k$}
\State $Gs \gets \bigcup_{m=1}^k \{ G ~|~ MG \reltwo_m G \}$
\Return $\bigcup_{G \in Gs} G$
\EndProcedure
\State
\Procedure{Concretize}{$M,k$}
\Return $\{ ps = \Call{AllProductions}{M,k}, start=M.\mathsf{start} \}$
\EndProcedure
\State
\Procedure{RuleBodyConstraints}{$rb$}
\Match{$rb$}
\Case{if $b$ then $p$, $rb'$}{}
\State $(\mathsf{ProductionIdentifier}(p) \Leftrightarrow b) \wedge \Call{RuleBodyConstraints}{rb'}$
\EndCase
\Case{$\cdot$}{$\top$}
\EndCase
\EndMatch
\EndProcedure
\State
\Procedure{Constraints}{$MG,k$}
\Match{$MG$}
\Case{\lstinline{exists} $id; MG'$}{ $\exists id \in \{0\ldots k-1\} .$\Call{Constraints}{$MG',k$}}
\EndCase
\Case{$i_1\{i_2\}$ \lstinline{-> exists} $i_3.rb;
  MG'$}{ $\Call{Constraints}{MG',k} \wedge$}
\State $\forall i_2 \in \{0\ldots k-1\}.\exists i_3 \in \{0\ldots k-1\}.\Call{RuleBodyConstraints}{rb}$
\EndCase
\Case{$\start\ \nt{e}$}{ $\top$}
\EndCase
\EndMatch
\EndProcedure
\State
\Procedure{InduceGrammar}{$MG,Ex^+,Ex^-,k$}
\State $G \gets \Call{Concretize}{M,k}$
\State $\varphi \gets \Call{Constraints}{M,k}$
\State $\varphi^+ \gets \bigwedge_{ex \in Ex^+} \mathsf{SemiringParse}(G,ex)$
\State $\varphi^- \gets \bigwedge_{ex \in Ex^+} \neg \mathsf{SemiringParse}(G,ex)$
\State $ret \gets \Call{SMT}{\varphi^+ \wedge \varphi^- \wedge \varphi}$
\Match{$ret$}
\Case{Some($vs$)}{ \textbf{return }$MG[vs]$}
\EndCase
\Case{None}{ \textbf{return }$\Call{InduceGrammar}{MG,Ex^+,Ex^-,k+1}$}
\EndCase
\EndMatch
\EndProcedure
\end{algorithmic}
\caption{Full Syntax-Guided Grammar Induction Algorithm}
\label{alg:grammarinduction-full}
\end{algorithm}

\newcounter{origsection}
\newcounter{origtheorem}

\newcommand{\forcenumber}[2]{%
  \setcounter{origsection}{\value{section}}%
  \setcounter{origtheorem}{\value{theorem}}
  \def\thesection{\arabic{section}}%
  \setcounter{section}{#1}%
  \setcounter{theorem}{#2}%
  \addtocounter{theorem}{-1}}

\newcommand{\resetnumber}{%
  \def\thesection{\Alph{section}}%
  \setcounter{section}{\value{origsection}}%
  \setcounter{theorem}{\value{origtheorem}}}

\setlist[itemize]{topsep=0.5em,itemsep=0.375em,partopsep=0pt,parsep=2pt}

\begin{theorem}[Correctness]
  If there is a Grammar $G$ such that $MG \reltwo G$, where $Ex^+ \subseteq \mathcal{L}(G)$ and $Ex^- \subseteq \overline{\mathcal{L}(G)}$, then $\Call{InduceGrammar}{MG,Ex^+,Ex^-,0}$ will return such a grammar.
  \begin{proof}
    This follows directly from Soundness (Theorem~\ref{thm:soundness}, stated and proven in Section~\ref{subsec:soundness}) and Completeness (Theorem~\ref{thm:completeness}, stated and proven in Section~\ref{subsec:completeness}).
  \end{proof}
\end{theorem}

\subsection{Soundness of \textsc{InduceGrammar}}
\label{subsec:soundness}

\begin{lemma}
  \label{lem:soundpos}
  If $vs \models \mathsf{SemiringParse}(\Call{Concretize}{MG,k},ex)$ then
  $ex \in \mathcal{L}(MG[vs])$.
\end{lemma}
\begin{proof}
  By prior work in Semiring Parsing~\cite{semiringparsing},
  $\mathsf{SemiringParse}(\Call{Concretize}{MG,k},ex)$ returns a formula
  $\varphi$ such that $\varphi \equiv \Sigma_{\mathit{parse} \in
    \mathit{parses}}\Pi_{p \in parse}\mathsf{ProductionIdentifier(p)}$.
  Intuitively, it returns an element in the semiring equivalent to the sum of
  all parse trees, where a parse tree is identified by the product of all rules
  used in the derivation.

  Our semiring is the conditional table semiring, so $\varphi \equiv
  \Sigma_{\mathit{parse} \in \mathit{parses}}\Pi_{p \in
    parse}\mathsf{ProductionIdentifier(p)}$ means that $\varphi \Leftrightarrow
  \bigvee_{\mathit{parse} \in \mathit{parses}}\bigwedge_{p \in
    parse}\mathsf{ProductionIdentifier(p)}$. If $vs \models \varphi$, then it is
  also the case that $vs \models \bigvee_{\mathit{parse} \in
    \mathit{parses}}\bigwedge_{p \in parse}\mathsf{ProductionIdentifier(p)}$.
  Thus, there exists some parse $\mathit{parse}$ such that $\forall p \in
  \mathit{parse}$, $p \in MG[vs]$. Thus, $ex \in \mathcal{L}(MG[vs])$.
\end{proof}

\begin{lemma}
  \label{lem:soundneg}
  If $vs \models \neg\mathsf{SemiringParse}(\Call{Concretize}{MG,k},ex)$ then
  $ex \not\in \mathcal{L}(MG[vs])$.
\end{lemma}
\begin{proof}
  By prior work in Semiring Parsing~\cite{semiringparsing},
  $\mathsf{SemiringParse}(\Call{Concretize}{MG,k},ex)$ returns a formula
  $\varphi$ such that $\varphi \equiv \Sigma_{\mathit{parse} \in
    \mathit{parses}}\Pi_{p \in parse}\mathsf{ProductionIdentifier(p)}$.
  Intuitively, it returns an element in the semiring equivalent to the sum of
  all parse trees, where a parse tree is identified by the product of all rules
  used in the derivation.

  Our semiring is the conditional table semiring, so $\varphi \equiv
  \Sigma_{\mathit{parse} \in \mathit{parses}}\Pi_{p \in
    parse}\mathsf{ProductionIdentifier(p)}$ means that $\varphi \Leftrightarrow
  \bigvee_{\mathit{parse} \in \mathit{parses}}\bigwedge_{p \in
    parse}\mathsf{ProductionIdentifier(p)}$. If $vs \models \neg\varphi$, then
  it is also the case that $vs \models \neg(\bigvee_{\mathit{parse} \in
    \mathit{parses}}\bigwedge_{p \in parse}\mathsf{ProductionIdentifier(p)})$.
  Thus, there does not exists any parse $\mathit{parse}$ such that $\forall p
  \in \mathit{parse}$, $p \in MG[vs]$. Thus, $ex \not\in \mathcal{L}(MG[vs])$.
\end{proof}

\begin{lemma}
  \label{lem:rbsound}
  Let $vs \models \Call{RuleBodyConstraints}{rb}$ and let $rb$ have no free
  variables. Then $\rb \reltwo \rb[vs]$.
  \begin{proof}
    This comes from straightforward induction on the rule body constraints,
    alongside straightforward evaluation for the True and False cases.
  \end{proof}
\end{lemma}

\begin{lemma}
  \label{lem:soundstep}
  If $vs \models \Call{Constraints}{MG,k}$ and $MG$ has no free variables, then
  $MG \reltwo_k MG[vs]$.
\end{lemma}
\begin{proof}
  By induction on $MG$.

  \paragraph*{Base Case: MG = \texttt{start id}$\{e\}$} We have $MG[vs] =
  \texttt{start } id'$ where $id' = \idmap{id\{n\}}$ and $e \equiv n$. Thus
  \begin{mathpar}
    \relationRule{} {
      e \equiv n\\
      \id' = \idmap{\nt{n}}
    } {
      \start\ \nt{e} \reltwo_k \start\ \id'
    }
  \end{mathpar}

  \paragraph*{Induction Case: MG = \texttt{exists id; }$MG'$} We have $MG[vs] =
  MG'[vs]$. By induction, $MG'[n/\texttt{id}] \reltwo_k MG'[vs] = MG[vs]$. Because $vs
  \models \Call{Constraints}{MG,k}$, we know $vs[\texttt{id}] \in \{0..k-1\}$.
  Thus
  \begin{mathpar}
  \relationRule{} {
    vs[\id] < m\\
    MG[n/\id] \reltwo_k MG[vs]
  } {
    \exts\ \id;\ MG' \rel MG[vs]
  }
  \end{mathpar}

  \paragraph*{Induction Case: $MG = i_1\{i_2\} \rightarrow \exts\ i_3.\rb;
    MG'$} By induction, $MG' \reltwo_k MG'[vs]$. Because $vs \models \forall i_2
  \in \{0\ldots k-1\}.\exists i_3 \in \{0\ldots
  k-1\}.\Call{RuleBodyConstraints}{rb}$, we know that for each $n \in 0\ldots
  k-1$ there exists $vs[i_{3,n}]$ such that, by Lemma~\ref{lem:rbsound},
  $\rb[vs[i_{3,n}]/i_3] \reltwo \rb[n/i_2][vs]$. Thus
  \begin{mathpar}
    \relationRule{} {
      MG' \reltwo_k MG'[vs]\\\\
      \forall n \in \{1\ldots k-1\}. \rb[n/i_2][vs[i_{3,n}]/i_3] \reltwo \rb[n/i_2][vs]\\\\
      \forall n \in \{1\ldots k-1\}. {i_1}_n = \idmap{i_1\{n\}}\\
    } {
      \id_1\{id_2\} \rightarrow \exts\ \id_3.\rb; MG \reltwo_k \id_1 \rightarrow rb[0/i_2][vs]; \dots; \id_{k-1} \rightarrow rb[k-1/i_2][vs]; MG'[vs]
    }
  \end{mathpar}
  Thus, as $\id_1 \rightarrow rb[0/i_2][vs]; \dots; \id_{k-1} \rightarrow
  rb[k-1/i_2][vs]; MG'[vs] = MG[vs]$ we know that $i_1\{i_2\} \rightarrow \exts\ i_3.\rb;
  MG' \reltwo_k MG[vs]$, as desired.

\end{proof}
\begin{lemma}
  \label{lem:soundnessk}
  If $G=\Call{InduceGrammar}{MG,Ex^+,Ex^-,k}$, and $G$ is not returned via a
  recursive call, then $MG \reltwo_k G$ and $Ex^+ \subseteq \mathcal{L}(G)$ and
  $Ex^- \subseteq \overline{\mathcal{L}(G)}$.
\end{lemma}
\begin{proof}
  If $G$ is not returned via a recursive call, then $G = MG[vs]$ where $vs
  \models \varphi^+ \wedge \varphi^- \wedge \varphi$ where $\varphi =
  \Call{Constraints}{M,k}$ and $\varphi^+ = \bigwedge_{ex \in
    Ex^+}\mathsf{SemiringParse}(\Call{Concretize}{MG,k},ex)$ and $\varphi^- = \bigwedge_{ex \in
    Ex^-}\neg\mathsf{SemiringParse}(\Call{Concretize}{MG,k},ex)$.

  As $vs \models \varphi$, from Lemma~\ref{lem:soundnessk}, we know that $MG
  \reltwo_k G$.

  Let $ex \in Ex^+$. As $vs \models
  \mathsf{SemiringParse}(\Call{Concretize}{MG,k},ex)$, from
  Lemma~\ref{lem:soundpos} we know that $ex \in \mathcal{L}(MG[vs])$. Thus,
  $Ex^+ \subseteq \mathcal{L}(MG[vs])$.

  Let $ex \in Ex^-$. As $vs \models
  \neg\mathsf{SemiringParse}(\Call{Concretize}{MG,k},ex)$, from
  Lemma~\ref{lem:soundpos} we know that $ex \not\in \mathcal{L}(MG[vs])$. Thus,
  $Ex^- \subseteq \overline{\mathcal{L}(MG[vs])}$.
\end{proof}

\begin{theorem}[Soundness]
  \label{thm:soundness}
  If $G=\Call{InduceGrammar}{MG,Ex^+,Ex^-,0}$, then $G \in \SemanticsOf{MG}$ and
  $Ex^+ \subseteq \mathcal{L}(G)$ and $Ex^- \subseteq \overline{\mathcal{L}(G)}$.
\end{theorem}
\begin{proof}
  Directly from Lemma~\ref{lem:soundnessk}.
\end{proof}

\subsection{Completeness of \textsc{InduceGrammar}}
\label{subsec:completeness}

\begin{lemma}
  \label{lem:completeassign}
  If $MG \reltwo_k G$, then there exists a model $vs$ such that $vs \models
  \Call{Constraints}{MG,k}$.
\end{lemma}
\begin{proof}
  By induction over the derivation of $MG \reltwo_k G$.

  \paragraph*{Base case: $i{e}\texttt{ start} \reltwo_k \idmap{i{e}}\texttt{
      start}$}. Vacuously satisfies constraints.

  \paragraph*{Induction Case: $MG = \texttt{exists id; }MG' \reltwo_k G$}. By
  induction, there exists $vs$ from $MG'[n/id]$ such that $vs \models
  \Call{Constraints}{MG',k}$. Thus, consider $vs[id \mapsto in]$. This satisfies
  $\exists n.\Call{Constraints}{MG',k}$, as desired.

  \paragraph*{Induction Case: $MG = i_1\{i_2\} \rightarrow \exts\ i_3.\rb; MG'
    \reltwo_k G$}. By induction we know that there exists $vs$ from $MG'$ such
  that $vs \models \Call{Constraints}{MG',k}$. For every $n \in \{1\ldots k-1\}$
  we know that there exists $i_{3_n}$ such that $\rb[n/i_2,i_{3_n}/i_3]$. Let
  $vs' = vs \cup \{i_{3_n} ~|~ n \in \{1\ldots k-1\}\}$. By simple evaluation,
  we know that $vs'\models \Call{RuleBodyConstraints}{rb}$.
\end{proof}

\begin{lemma}
  \label{lem:completeneg}
  If $G = MG[vs]$ and $ex \not\in \mathcal{L}(G)$ then $vs \models \neg\mathsf{SemiringParse}(\Call{Concretize}{MG,k},ex)$.
\end{lemma}
\begin{proof}
  Assume not. Then $vs \not\models
  \mathsf{SemiringParse}(\Call{Concretize}{MG,k},ex)$. Our semiring is the
  conditional table semiring, so
  $\mathsf{SemiringParse}(\Call{Concretize}{MG,k},ex)$ is equivalent to the
  formula $\Sigma_{\mathit{parse} \in \mathit{parses}}\Pi_{p \in
    parse}\mathsf{ProductionIdentifier(p)}$. This means that it is not the case
  that $vs$ models $\bigvee_{\mathit{parse} \in \mathit{parses}}\bigwedge_{p \in
    parse}\mathsf{ProductionIdentifier(p)}$. This means there is a ruleset in
  $vs$, namely $G$, that permits a parsing of $ex$. But, $vs \not\in
  \mathcal{L}(G)$. Contradiction.
\end{proof}

\begin{lemma}
  \label{lem:completepos}
  If $G = MG[vs]$ and $ex \not\in \mathcal{L}(G)$ then $vs \models \mathsf{SemiringParse}(\Call{Concretize}{MG,k},ex)$.
\end{lemma}
\begin{proof}
  Assume not. Then $vs \not\models
  \neg\mathsf{SemiringParse}(\Call{Concretize}{MG,k},ex)$. Our semiring is the
  conditional table semiring, so
  $\mathsf{SemiringParse}(\Call{Concretize}{MG,k},ex)$ is equivalent to the
  formula $\Sigma_{\mathit{parse} \in \mathit{parses}}\Pi_{p \in
    parse}\mathsf{ProductionIdentifier(p)}$. This means it is not the case that
  $vs$ models $(\bigvee_{\mathit{parse} \in \mathit{parses}}\bigwedge_{p \in
    parse}\mathsf{ProductionIdentifier(p)})$. This means that there are no
  productions in $MG$ that permit parsing. But, $vs \in \mathcal{L}(G) \subseteq
  \mathcal{L}(G)$. Contradiction.
\end{proof}

\begin{theorem}
  \label{lem:kcompleteness}
  If there exists some $G$ such that $MG \reltwo_k G$ and $Ex^+ \subseteq
  \mathcal{L}(G)$ and $Ex^- \subseteq \overline{\mathcal{L}(G)}$ then
  $\Call{InduceGrammar}{MG,Ex^+,Ex^-,k}$ will return without a recursive call.
  \begin{proof}
    By Lemma~\ref{lem:completeassign}, we know there exists $vs$ such that $vs
    \models \Call{Constraints}{MG,k}$. By Lemma~\ref{lem:completepos}, we know
    that for each $ex \in Ex^+$, $vs \models
    \mathsf{SemiringParse}(\Call{Concretize}{MG,k},ex)$. By
    Lemma~\ref{lem:completeneg}, we know that for each $ex \in Ex^-$, $vs \models
    \neg\mathsf{SemiringParse}(\Call{Concretize}{MG,k},ex)$. Thus, there exists
    some $vs$ that satisfies $\varphi^+ \wedge \varphi^- \wedge \varphi$. Thus,
    $\Call{InduceGrammar}{MG,Ex^+,Ex^-,k}$ will return some such $MG[vs]$.
  \end{proof}
\end{theorem}

\begin{theorem}[Completeness]
\label{thm:completeness}
If there exists some $G$ such that $G \in \SemanticsOf{MG}$ and $Ex^+ \subseteq
\mathcal{L}(G)$ and $Ex^- \subseteq \overline{\mathcal{L}(G)}$ then
$\Call{InduceGrammar}{MG,Ex^+,Ex^-,0}$ will return.

Lemma~\ref{lem:completepos}.
\end{theorem}
\begin{proof}
  Directly from Lemma~\ref{lem:kcompleteness}.
\end{proof}

\clearpage

\section{Case Studies}
\label{sec:case-studies}
Here, we describe our experience implementing metagrammars for two more particularly complex formats: XML and SQL. In the first, we demonstrate that a major benefit of \Sys is the grammar engineer's ability to control the search space, drastically reducing the number of possibilities our induction algorithm must consider. In the second case, we show how \Sys can be used to distinguish between SQL versions and thereby give useful, automated feedback to the query writer.

\subsection{XML Metagrammar Development}

\begin{figure}[!h]\small%
\begin{boxed-listing}[left=-28pt]{}
exists numElementTypes : nat.

Element{i} ->
  ? Header{i} Text Footer{i}
  ? Header{i} Container{i} Footer{i}.

Container{i} ->
  ? ""
  [? Element{j} Elements{i} for j = 0 to numElementTypes].
  
S -> [? Element{i} for i = 0 to numElementTypes].
start S.
\end{boxed-listing}
\caption{Simplified XML metagrammar.}%
\label{fig:xml-no-constraints} 
\end{figure}

\begin{figure}[!h]\small%
\begin{boxed-listing}[left=-28pt]{}
<class>
  <name>Intro to Computer Science</name>
  <students>
    <student>Jane Doe</student>
    <student>John Public</student>
  </students>
</class>
\end{boxed-listing}
\caption{Example XML document describing an introduction to computer science class.}
\label{fig:csxml}
\end{figure} 

XML is a general-purpose document class often used for data storage and data transfer.
Figure~\ref{fig:xml-no-constraints} presents a simplified version of our XML metagrammar and Figure~\ref{fig:csxml} presents
an XML document with which to specialize the metagrammar. 

The metagrammar itself permits several different types of \emph{elements}, each of which can be a \emph{data element}, a \emph{container element}, or both.
A \emph{data element} is an XML element that contains some text before its closing tag. A \emph{container element} is an XML element that contains some XML before its closing tag. 
For expository purposes, Figure~\ref{fig:xml-no-constraints} specifies only
these two types of XML elements, and does not include additional features such as XML attributes, though the metagrammar we defined and used in our experiments does.  

Our example XML document contains four different XML elements: \lstinline{class}, \lstinline{name}, \lstinline{students}, and \lstinline{student}.
The \lstinline{class} and \lstinline{students} elements are container elements while
the \lstinline{name} and \lstinline{student} elements are data elements.

Unfortunately, this simple XML meta-grammar contains a lot of redundancy. 
For instance, consider the grammars $G_1$: 
\begin{lstlisting}
S -> Element{0}

Element{0} -> "<a>" Container{0} "</a>"
Element{1} -> "<b>" Text "</b>"

Container{0} -> Element{1}
\end{lstlisting}
and $G_2$:
\begin{lstlisting}
S -> Element{0}

Element{0} -> "<a>" Container{0} "</a>"
Element{2} -> "<b>" Text "</b>"

Container{0} -> Element{2}
\end{lstlisting}
These two grammars are clearly semantically equivalent, but syntactically distinct.
The presence of syntactic redundancy in a search space is a well-known problem in program synthesis in general --- in any SyGuS system, such redundancy enlarges the
search space and slows down synthesis. \Sys is no exception: for sufficiently complex examples with large amounts of redundancy, our Max-SMT calls are prohibitively slow.

\begin{figure}[!h]\small%
\begin{boxed-listing}[left=-28pt]{}
exists numElementTypes : nat.

Element{i : nat} ->
  ? Header{i} Text Footer{i}
  ? Header{i} Container{i} Footer{i}

Container{i : nat} ->
  ? ""
  [? Element{j} Elements{i} for j = 0 to numElementTypes].

forall (i:nat) : constraint
    (|Productions(Element{i})| > 0) => (|Productions(Element{i-1})| > 0).
start Element{0}.
\end{boxed-listing}
\caption{Simplified XML metagrammar.}%
\label{fig:xml-with-constraints}
\end{figure}

Fortunately, our metagrammar language affords a solution: the programmer can add constraints to cut the search space back down. In Figure~\ref{fig:xml-with-constraints}, we have added constraints to ensure that if \texttt{Element\{i\}} has a nonzero number of productions, then \texttt{Element\{i-1\}} must \emph{also} have a nonzero number of productions; we further require that the first element must be \texttt{Element\{0\}}. 

As a result, the singleton grammar with \texttt{Element\{1\}} is no longer a valid grammar. Grammars which "skip" \texttt{Element\{1\}} in favor of higher-indexed \texttt{Element}s are similarly disallowed.
\OMIT{ 
Furthermore, if \texttt{Element\{0\}} must contain exactly one text element within it, and no other elements are used, that element will always be \texttt{Element\{1\}}. Because \texttt{Element\{1\}} has no productions, \texttt{Element\{2\}} cannot be the single text element within \texttt{Element\{0\}}.
}
This is a tactic that can be applied to general recursive data -- when constructing a metagrammar with multiple mutually recursive elements, providing a default ordering in which those elements are used helps constrain the search space. In our case, \texttt{Element\{0\}} gets used first, followed by \texttt{Element\{1\}}, and so on.

\begin{figure}
    \includegraphics[scale=.84]{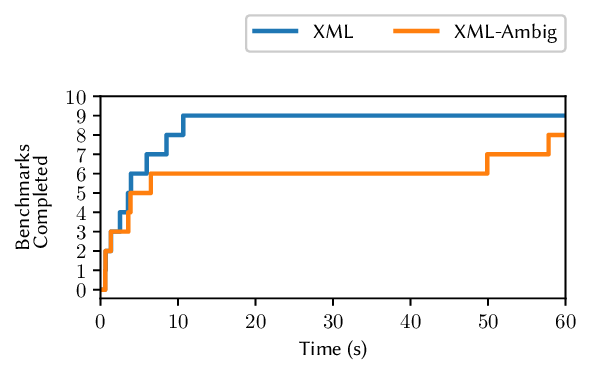}
    \caption{Number of benchmarks that terminate in a given time on XML and XML-Ambig.}
    \label{fig:runtime-diffs}
\end{figure}

We demonstrate the usefulness of this technique by running \Sys on both the ambiguous and non-ambiguous XML metagrammars; the results are summarized in Figure~\ref{fig:runtime-diffs}. By minimizing the ambiguity of the metagrammar, XML inference becomes much more tractable without reducing expressiveness. Although synthesis still failed on one benchmark (task 7, involving 20 PEs and 20 NEs), we found that the time taken per additional example increased at a substantially slower rate in the less-ambiguous metagrammar.

\paragraph*{Takeaways} One of the strengths of \Sys lies in the ability of the metagrammar developer to control the search space. By refactoring the metagrammar, the developer can encode the same grammatical domain in a way that is more amenable to synthesis. Doing so provides fundamental advantages over general-purpose, unguided grammatical inference techniques where the search space is overwhelmingly large and unconstrained.


\subsection{SQL Identifiers}
The relational database management ecosystem consists of
players such as Amazon AWS, Google BigQuery and Microsoft Azure competing with one another to
attract new users. While all these systems provide support for user queries, the Structured Query Language (or SQL)
grammatical domain has languages that differ between companies or even within the same company. For example, 
Google BigQuery has recently changed from BigQuery SQL (now referred to in the docs as LegacySQL) 
to Google Standard SQL. One difference between SQL dialects is how
languages handle identifiers. Consider the following SELECT statements taken from the Google Cloud reference
on migrating to Google Standard SQL \cite{google_standard_sql}:
\begin{figure}[!h]\small%
\begin{boxed-listing}[left=-28pt]{}
    1. SELECT word FROM [bigquery-public-data:samples.shakespeare] LIMIT 1;
    2. SELECT word FROM `bigquery-public-data.samples.shakespeare` LIMIT 1;
    3. SELECT COUNT(\*) AS `rows` FROM `bigquery-public-data.samples.shakespeare`;
\end{boxed-listing}
\caption{Three SQL select statements taken from Google cloud documentation.}
\label{fig:google-sql-selects}
\end{figure} 
The first example is written in the LegacySQL dialect whereas the latter two examples are written in the
preferred Google Standard SQL dialect. To differentiate the dialects, we constructed a metagrammar for SQL
SELECT statements which includes the snippet below. The metagrammar includes an emit statement that gives feedback
when the SQL example uses a production that belongs to LegacySQL and not Google Standard SQL. \newpage

\begin{figure}[!h]\small%
\begin{boxed-listing}[left=-28pt]{}
Identifier ->
  ? IdentifierChars
  ? Identifier "." Identifier
  ? "`" EscapedIdentifierChars "`" named BacktickIdentifier
  ? "[" EscapedIdentifierChars "]" named BracketIdentifier.

emit "Square bracket identifier is indiciative of Legacy SQL. 
  Please switch to backticks (``)" if BracketIdentifier.

\end{boxed-listing}
\caption{Candidate productions in our SQL SELECT Statement metagrammar.}
\label{fig:sql-identifiers}
\end{figure}
With our metagrammar specification, \Sys identifies that the first SELECT statement in Fig. \ref{fig:google-sql-selects} uses the 
square bracket identifier whereas the latter two recognize the backtick identifier. We also note that \Sys is able to handle
cases where a simpler program, like one that just searches for square brackets, would fail. The example below from
the Google Cloud SQL docs uses square brackets but does not use square brackets for identifiers. Our tool does not mark it
as using the BracketIdentifier production.
\begin{figure}[!h]\small%
\begin{boxed-listing}[left=-28pt]{}
WITH T AS (
  SELECT x FROM UNNEST([1, 2, 3, 4]) AS x
),
TPlusOne AS (
  SELECT x + 1 AS y FROM T
),
TPlusOneTimesTwo AS (
  SELECT y * 2 AS z FROM TPlusOne
)
SELECT z
FROM TPlusOneTimesTwo;
\end{boxed-listing}
\caption{Google SQL statement with square brackets which are not for identifiers.}
\label{fig:google-sql-unnest}
\end{figure} 

In this case, \Sys could be used to aid in an internal migration from Google LegacySQL 
to Google Standard SQL by identifying scripts that are using the older syntax and flagging them. 
Additionally, automatically identifying SQL versions could be useful for a software engineer whose job includes debugging client scripts. 
Based on a discussion with a colleague who works on Google BigQuery, clients that move from one relational database service to another will sometimes 
copy old scripts directly into the new engine without any adjustment. Instead of an error message saying a certain character
was unexpected, one could generate a dialect detector using \Sys that flags scripts using a different version of SQL. The feedback could be propagated to the query engine
user directly with some information how to update their script or to the software engineers whose job it is to debug failing customer scripts. Consider the three 
SELECT statements below taken from the 2022 Microsoft Server SQL documentation \cite{msf_server_sql}:
\begin{figure}[!h]\small%
\begin{boxed-listing}[left=-28pt]{}
    1. SELECT * FROM HumanResources.Employee WHERE NationalIDNumber = 153479919
    2. SELECT * FROM [HumanResources].[Employee] 
         WHERE [NationalIDNumber] = 153479919
    3. SELECT * FROM [SalesOrderDetail Table] WHERE [Order] = 10;
\end{boxed-listing}
\caption{Three SQL select statements taken from Microsoft Server SQL documentation.}
\label{fig:msft-sql-selects}
\end{figure}
The first example above has no extra quoting for the identifiers since none is needed, 
the second uses the square bracket formatting even though it is not needed, and the third uses
square bracket formatting when it is required. \Sys recovers the square bracket identifier production for
the latter two examples but not the first, as expected.
\fi

\end{document}